\documentclass[11pt]{article}
\usepackage[utf8]{inputenc}
\usepackage[T1]{fontenc}
\usepackage{lmodern}

\usepackage{amsmath}
\usepackage{mathtools}
\usepackage{amssymb}
\usepackage{amsthm}

\usepackage{array}
\usepackage{comment}

\usepackage[plainpages=false,pdfpagelabels,colorlinks=true,citecolor=blue,hypertexnames=false]{hyperref}

\newtheorem{thm}{Theorem}

\newtheorem{lem}{Lemma}

\def\N{\mathbb{N}}
\def\Z{\mathbb{Z}}
\def\C{\mathbb{C}}
\def\F{\mathbb{F}}
\def\K{\mathbb{K}}
\def\M{\mathsf{M}}
\def\MM{\mathsf{MM}}

\def\Ff{\mathsf{F}}
\def\T{\mathsf{T}}
\def\E{\mathsf{E}}
\def\SS{\mathsf{S}}

\usepackage{xcolor}    

\usepackage[textwidth=1.7cm]{todonotes}

\title{A Simple and Fast Algorithm  for Computing \\ the $N$-th Term of a Linearly Recurrent Sequence}
 \author{\href{https://specfun.inria.fr/bostan/}{Alin Bostan}$^\star$ and \href{{https://q.c.titech.ac.jp/mori/}}{Ryuhei Mori}$^\dagger$
\\
$^\star${Inria, Palaiseau, France} and  $^\dagger${Tokyo Institute of Technology, Japan}                          
\\   
\href{mailto:alin.bostan@inria.fr}{\tt alin.bostan@inria.fr}, \; \href{mailto:mori@c.titech.ac.jp}{\tt mori@c.titech.ac.jp}
}                                 
\date{August 19, 2020}

\usepackage{graphicx}
\usepackage{algorithm}
\usepackage[noend]{algpseudocode}
                
\begin{document}

\maketitle    

\begin{abstract} 
We present a simple and fast algorithm for computing the $N$-th term of a given linearly recurrent sequence.
Our new algorithm uses $O(\M(d) \log N)$ arithmetic operations, where $d$ is the order of the recurrence,
and $\M(d)$ denotes the number of arithmetic operations for computing the product of two polynomials of degree~$d$.
The state-of-the-art algorithm, due to Charles~Fiduccia (1985), has the same arithmetic complexity up to a constant factor.
Our algorithm is simpler, faster and obtained by a totally different method.
We also discuss several algorithmic applications, notably to      
polynomial modular exponentiation,
powering of matrices and high-order lifting.
\end{abstract}

\noindent {\bf Keywords}: 
Algebraic Algorithms;
Computational Complexity;
Linearly Recurrent Sequence;
Rational Power Series;
Fast Fourier Transform

\section{Introduction}
                                         
\subsection{General context}
Computing efficiently selected terms in sequences is a basic and fundamental
algorithmic problem, whose applications are ubiquitous, for instance in 
theoretical computer science~\cite{RaEi04,LuZi08}, 
algebraic complexity theory~\cite{PaSt73,Strassen74}, 
computer algebra~\cite{Giesbrecht95,Storjohann02,Lecerf10},
cryptography~\cite{GoHa99,GoHaWu01,GiGo05,GHMV13}, 
algorithmic number theory~\cite{Strassen76,AKS04},
effective algebraic geometry~\cite{BoGaSc07,Harvey14},
numerical analysis~\cite{Mezzarobba10,NaPa13} 
and computational biology~\cite{NuDu13}.

\smallskip
In simple terms, the problem can be formulated as follows:
  
\bigskip
\vspace{20pt}
\hrule
\vspace{0pt}
\begin{quote}
	Given a sequence $(u_n)_{n \geq 0}$ in an effective ring\footnote[0]{The ring $R$ is assumed to be commutative with unity and \emph{effective} in the sense that its elements are represented using some data structure, and there exist algorithms for performing the basic ring operations $(+,-,\times)$ and for testing equality of elements in~$R$.}
 $R$, and given a positive integer $N\in\mathbb{N}$, compute the term $u_N$ as fast as possible.
\end{quote}         
\vspace{0pt}
\hrule
\vspace{10pt}

Here, the input $(u_n)_{n \geq 0} \in R^\N$ is assumed to be a recurrent
sequence, specified by a data structure consisting in a \emph{recurrence
relation} and sufficiently many \emph{initial terms} that uniquely determine
the sequence $(u_n)_{n \geq 0}$.

Efficiency is measured in terms of \emph{ring operations} (algebraic model),
or of \emph{bit operations} (Turing machine model).   
The cost of an algorithm is respectively estimated in terms of 
\emph{arithmetic complexity} or of \emph{binary complexity}.
Both measures have their
own usefulness: the algebraic model is relevant when ring operations have
essentially unit cost (typically, when $R$ is a finite ring such as the prime
field~$\F_p := \Z/p \Z$), while the bit complexity model is relevant when
elements of~$R$ have a variable bitsize, and thus ring operations in $R$ have
variable cost (typically, when $R$ is the ring $\Z$ of integer numbers).

\smallskip The recurrence relation satisfied by the input sequence $(u_n)_{n\geq 0}$ might be of several types:     

\begin{itemize}
	\item[(C)] \emph{linear with constant coefficients}, that is of the form,
\[u_{n+d}=c_{d-1}u_{n+d-1}+\dots+c_0u_n,\qquad n \ge 0,\]
for some given coefficients $c_0, \ldots, c_{d-1}$ in $R$.	
In this case we simply say that the sequence is \emph{linearly recurrent} (or, \emph{C-recursive}). The most basic examples are the \textcolor{magenta}{\href{https://oeis.org/A051128}{geometric}} sequence $(q^n)_{n \geq 0}$, for $q\in R$, and the 	
{\href{https://oeis.org/A000045}{Fibonacci}} sequence $(F_n)_{n}$ with $F_{n+2} = F_{n+1} + F_n$ 
and $F_0=0, F_1=1$.      
	\item[(P)] \emph{linear with polynomial coefficients}, that is of the form,
\[u_{n+d}=c_{d-1}(n) u_{n+d-1}+\dots+c_0(n) u_n,\qquad n \ge 0,\]
for some given rational functions $c_0(x), \ldots, c_{d-1}(x)$ in $R(x)$.	
In this case the sequence is called \emph{holonomic} (or, \emph{P-recursive}). Among the most basic examples, other than the C-recursive ones, there is the \textcolor{magenta}{\href{https://oeis.org/A000142}{factorial}} sequence  $(n!)_{n \geq 0} = (1, 1, 2, 6, 24, 120, \ldots)$
and the 	  
\textcolor{magenta}{\href{https://oeis.org/A001006}{Motzkin}} sequence
$(u_n)_{n \geq 0} = (1, 1, 2, 4, 9, 21, 51, \ldots)$ specified by the recurrence $u_{n+1}=\frac{2n+3}{n+3} \cdot u_{n} + \frac{3n}{n+3} \cdot u_{n-1}$ and the initial conditions $u_0=u_1=1$.   
	\item[(Q)] \emph{linear with polynomial coefficients in $q$ and $q^n$}, that is of the form,
\[u_{n+d}=c_{d-1}(q, q^n) u_{n+d-1}+\dots+c_0(q, q^n) u_n,\qquad n \ge 0,\]
for some $q\in R$ and some rational functions $c_0(x,y), \ldots, c_{d-1}(x,y)$ in $R(x,y)$.	
In this case, the sequence is called \emph{$q$-holonomic}; such a sequence can be seen as a $q$-deformation of a holonomic sequence (in the sense that when $q \mapsto 1$, the limit sequence tends to be holonomic).
A typical example is the \textcolor{magenta}{\href{https://oeis.org/A069777}{$q$-factorial}} $
[n]_q!:=  (1+q) \! \cdots \! (1+q+\cdots +q^{n-1})$.
\end{itemize}

In all these classes of examples, the recurrence is \emph{linear}, and the
integer~$d$ that governs the length of the recurrence relation is called the
\emph{order} of the corresponding linear recurrence.

Of course, some interesting sequences satisfy \emph{nonlinear} recurrences, as is the case for the so-called              
\textcolor{magenta}{\href{https://oeis.org/A006720}{Somos-4}} sequence $(1, 1, 1, 1, 2, 3, 7, 23, 59,\ldots)$ defined by: 
$u_{n+4}=  \bigl(u_{n+3} u_{n+1}+u_{n+2}^2\bigr)/u_n$
together with $u_0=\cdots =u_3=1$, 
but we will not consider this larger class in what follows. 

\medskip For computing the $N$-th term $u_N$ in a sequence of type (P), resp.~(Q), the best known algorithms are presented in~\cite{ChCh88,BoGaSc07},
resp.~in~\cite{Bostan20}. In the algebraic model, they rely on an
algorithmic technique called \emph{baby-step / giant-step}, which allows
to compute~$u_N$ using a number of operations in~$R$ that is almost linear
in~$\sqrt{N}$, up to logarithmic factors. This should be contrasted with the
direct iterative algorithm, of arithmetic complexity linear in~$N$.

In the bit model, the same references provide different algorithms based on a
different technique, called \emph{binary splitting}; these algorithms are
\emph{quasi-optimal} in the sense that they are able to compute~$u_N$ in a
number of bit operations almost linear (up to logarithmic factors) in
the bitsize of the output value~$u_N$. Once again, this should be contrasted
with the direct iterative algorithm, whose binary complexity is larger by at
least one order of magnitude (e.g., in case (P) the naive algorithm has bit
complexity $O(N^3)$).

\smallskip 
\subsection{The case of C-recursive sequences}

In what follows, we will restrict our attention to the case (C) only. This
case obviously is a subcase of both cases (P) and (Q). It presents an
exceptional feature with respect to the algebraic model: contrary to the
general cases (P) and (Q), in case (C) it is possible to compute the
term~$u_N$ using a number of arithmetic operations in~$R$ that is only
logarithmic in~${N}$. 

For the geometric sequence~$u_n = q^n$, this is known since Pingala
($\sim$200~BC) who seemingly is the inventor of the algorithmic method of
\emph{binary powering}, or \emph{square-and-multiply}~\cite[\S4.6.3]{Knuth81}.
The corresponding algorithm is recursive and based on the equalities
\[q^N=\begin{cases}(q^{N/2})^2,&\text{if $N$ is even,}\\
q\cdot (q^{\frac{N-1}{2}})^2,&\text{else.}\end{cases}\]   
The arithmetic complexity of this algorithm is bounded by $2\log N$ multiplications\footnote{In all this article, the notation $\log$ refers to the logarithm in base $2$.}
in~$R$, which represents a tremendous improvement compared to the naive
iterative algorithm that computes the term $q^N$ in $N-1$  multiplications in~$R$, by simply unrolling the
recurrence $u_{n+1} = q \cdot u_n$ with $q_0=1$.     

In the general case (C), Miller and Spencer Brown showed in 1966~\cite{MiBr66}
that a similar complexity can be obtained by converting\footnote{In~\cite[p.~74]{CeMiPi87}, the authors call this \emph{un truc bien connu} (``a well-known trick'').} the scalar recurrence of order~$d$
\begin{equation} \label{eq:rec}
u_{n+d}=c_{d-1}u_{n+d-1}+\dots+c_0u_n,\qquad n \ge 0,
\end{equation}    
into a vector recurrence of order 1
\begin{equation} \label{eq:rec-vect}
\underbrace{\begin{bmatrix}u_{n}\\ u_{n+1}\\\vdots \\u_{n+d-1}
\end{bmatrix}}_{v_{n}}
=\underbrace{\begin{bmatrix}&1&& \\ & & \ddots & \\ & & & 1 \\ c_0 & c_1 & \cdots & c_{d-1}
  \end{bmatrix}
  }_{M} 
\times
\underbrace{\begin{bmatrix}u_{n-1}\\u_{n}\\\vdots \\u_{n+d-2}
\end{bmatrix}}_{v_{n-1}},
\qquad n\ge1,
\end{equation}    
and by using binary powering in the ring $\mathcal{M}_d(R)$,
of $d \times d$ matrices with coefficients in $R$,
in order to compute $M^N$ recursively by
\[M^N=\begin{cases}(M^{N/2})^2,&\text{if $N$ is even,}\\
M\cdot (M^{\frac{N-1}{2}})^2,&\text{else.}\end{cases}\] 
From there, $u_N$ can be read off the matrix-vector product $v_N = M^N \cdot
v_0$.
The arithmetic complexity of this method is $O(d^\theta \log N)$ operations in~$R$, where
$\theta \in [2,3]$ is any feasible exponent for matrix multiplication in
$\mathcal{M}_d(R)$.

Strangely enough, the paper~\cite{MiBr66} of Miller and Spencer Brown was
largely overlooked in the subsequent literature, and their result has been
rediscovered several times in the~1980s. For instance, Shortt~\cite{Shortt78}
proposed a $O(\log N)$ algorithm for computing the $N$-th Fibonacci number\footnote{Shortt's algorithm had actually appeared before, in the 1969 edition of Knuth's book~\cite[p.~552]{Knuth69}, as a solution of Ex.~26 (p.~421, \S 4.6.3). The algorithm is based on the so-called \emph{doubling formulas}
$(F_{2n}, F_{2n-1}) = (F_n^2+2F_nF_{n-1}, F_n^2+F_{n-1}^2)$, actually due to Lucas (1876) and Catalan (1886), see e.g.~\cite[Ch.~XVII]{Dickson1919}.
The currently best implementation for computing $F_N$ over $\Z$ 
(\href{https://gmplib.org/manual/Fibonacci-Numbers-Algorithm}{mpz\_fib\_ui} from \href{https://gmplib.org}{GMP}) uses a variant of this method, 
requiring just two squares (and a few additions) per binary digit of~$N$.}$^{,}$\footnote{Already in 1899, G. de Rocquigny asked ``for an expeditious procedure to compute a very distant term of the Fibonacci sequence''~\cite{Rocquigny1899}.
In response, several methods (including the one mentioned by Knuth in~\cite[p.~552]{Knuth69}) have been published one year later by Rosace (alias), E.-B. Escott, E.~Malo, C.-A.~Laisant and G.~Picou~\cite{REMLP1900}.
This fact does not seem to have been noticed in the modern algorithmic literature before the current paper, although the reference~\cite{REMLP1900} is mentioned in Dickson's formidable book~\cite[p.~404]{Dickson1919}.
} 
and extended it together with Wilson~\cite{WiSh80} to the computation of
order-{$d$} {F}ibonacci numbers in $O(d^3 \log N)$ arithmetic operations. The
same cost has also been obtained by 
Dijkstra~\cite{Dijkstra79} and
Urbanek~\cite{Urbanek80}.
Pettorossi~\cite{Pettorossi80}, and independently Gries and
Levin~\cite{GrLe80}, improved the algorithm and lowered the cost 
to $O(d^2 \log N)$, essentially by taking into account the sparse structure of
the matrix~$M$.
See also~\cite{Er83,MaRe84,Er86,Er88,Holloway88,PrTa89,GiGo06,Khomovsky18} for similar algorithms.

\subsection{Fiduccia's algorithm}\label{sec:fiduccia}

The currently best algorithm is due to Fiduccia\footnote{The idea already appears in the 1982 conference paper~\cite{Fiduccia82}.
We have discovered that the same algorithm had been sketched by D.~Knuth in the corrections/changes to~\cite{Knuth81} published in 1981       in~\cite[p.~28]{Knuth81e}, where he attributes the result to R.~Brent.
Almost surely, C.~Fiduccia was not aware about this fact.}~\cite{Fiduccia85}. It is based
on the following observation: the matrix $M$ in~\eqref{eq:rec-vect} is the transpose of the companion
matrix $C$ which represents the $R$-linear multiplication-by-$x$ map from the
quotient ring $R[x]/(\Gamma)$ to itself, where $\Gamma = x^d-\sum_{i=0}^{d-1} c_i x^i$.
Therefore,
denoting by $e$ the row vector $e = \begin{small}\begin{bmatrix}1 & 0 & \cdots & 0 \end{bmatrix} \end{small}$,
the $N$-th term $u_N$ equals
\begin{equation}    \label{eq:fiduccia}
u_N 
=  e\cdot v_N  
= e\cdot M^N \cdot v_0
=  \left( C^N \cdot e^T \right)^T \cdot v_0 
= \langle x^N \bmod \Gamma, \; v_0  \rangle,
\end{equation}  
where the inner product takes place between the vector $v_0
=
\begin{bmatrix}
	u_{0} & \cdots & u_{d-1}
\end{bmatrix} 
$ of initial
terms of $(u_n)_{n \geq 0}$, and the vector whose entries are the coefficients
of the remainder $(x^N \bmod \Gamma)$ of the Euclidean division of
$x^N$ by $\Gamma$. 

Therefore, computing $u_N$ is reduced to computing the coefficients of $(x^N
\bmod \Gamma)$, and this can be performed efficiently by using binary powering in
the quotient ring $A := R[x]/(\Gamma)$, at the cost of $O(\log N)$ multiplications
in~$A$. Each multiplication in $A$ may be performed using $O(\M(d))$
operations in~$R$~\cite[Ch.~9, Corollary~9.7]{GaGe13}, where $\M(d)$ denotes
the arithmetic cost of polynomial multiplication in $R[x]$ in degree~$d$. 
Using Fast Fourier Transform (FFT) methods, one may take $\M(d) = O(d
\log d)$ when $R$ contains enough roots of unity, and $\M(d) = O(d \log d \log
\log d)$ in general~\cite[Ch.~8]{GaGe13}. 

In conclusion, Fiduccia's algorithm allows the computation of the $N$-th
term~$u_N$ of a linearly recurrent sequence of order $d$ using $O(\M (d) \log
N)$ operations in~$R$. Since 1985, this is the state-of-the-art algorithm for
this task in case~(C).

A closer inspection of the proof of \cite[Corollary~9.7]{GaGe13} shows that a more
precise estimate for the arithmetic cost of Fiduccia's algorithm is 
\begin{equation} \label{eq:cost-Fiduccia}
 \Ff(N,d) =        3 \, \M (d) \lfloor \log N\rfloor + O(d \log N)
\end{equation} 
operations in $R$. This comes from the fact that squaring\footnote{Note that multiplying by~$x$ in~$A$ is much easier and has linear arithmetic cost $O(d)$.} in $A = R[x]/(\Gamma)$
is based on one polynomial multiplication in degree less than $d$ followed by
an Euclidean division by $\Gamma$ of a polynomial of degree less than $2d$. The
Euclidean division is reduced to a power series division by the reversal
$Q(x) :=x^d \cdot \Gamma(1/x)$ of $\Gamma$, followed by a polynomial
multiplication in degree less than $d$. The reciprocal of $Q(x)$ is
precomputed modulo $x^{d}$ once and for all (using a formal Newton iteration)
in $\, 3\, \M(d) + O(d)$ operations in $R$, and then each squaring in $A$ also
takes $3 \, \M(d) + O(d)$ operations in~$R$. The announced cost
from~\eqref{eq:cost-Fiduccia} follows from the fact that binary powering uses
$\lfloor \log N\rfloor$ squarings and at most $\lfloor \log N\rfloor$ 
multiplications by~$x$.


\subsection{Main results}      \label{sec:main}

\medskip
We propose in this paper a new and simpler algorithm, with a better cost. More
precisely, our first main complexity result is:

\begin{thm} \label{thm:main} One can compute the $N$-th term of a
linearly recurrent sequence of order $d$ with coefficients in a ring $R$ using
\begin{equation*} 
	\T(N,d) = 
	2 \, \M(d) \lceil\log (N+1)\rceil + \M(d)  
\end{equation*} 
arithmetic operations in $R$. 
\end{thm}

The proof of this result is based on a very natural algorithm, which will be
presented in Section~\ref{sec:proof}. Let us remark that it improves by a
factor of~$1.5$ the complexity of Fiduccia's algorithm. 

This factor is even higher in the FFT setting, where polynomial multiplication
is assumed to be performed using Fast Fourier Transform techniques. In this
setting, we obtain the following complexity result, which will be proved in
Section~\ref{sec:FFT}.

\begin{thm} \label{thm:main-FFT} One can compute the $N$-th term of a
linearly recurrent sequence of order $d$ with coefficients in a field $\K$ supporting FFT using
\begin{equation*} 
	\sim \frac{2}{3} \, \M(d) \log (N)
\end{equation*} 
arithmetic operations in $\K$. 
\end{thm}     

Algorithms~\ref{alg:base} and~\ref{alg:base-rec}  
(underlying Theorem~\ref{thm:main})
and
Algorithm~\ref{alg:FFT} (underlying Theorem~\ref{thm:main-FFT})
 are both of LSB-first (least significant bit first)
type. This prevents them from computing simultaneously
\emph{several} consecutive terms of high indices, such as $u_N, \ldots,
u_{N+d-1}$. This makes a notable difference with Fiduccia's algorithm from
\S\ref{sec:fiduccia}.                          
For this reason, we will design a second algorithm, of   
MSB-first (most significant bit first) type, by ``transposing'' 
Algorithm~\ref{alg:base}. This leads to the following complexity result.

\begin{thm} \label{thm:main-t} One can compute the terms of indices $N-d+1, \ldots, N$ of a
linearly recurrent sequence of order $d$ with coefficients in a ring $R$ using
\begin{equation*} 
	2 \, \M(d)  \lceil\log (N+1)\rceil + O(\M(d))
\end{equation*} 
arithmetic operations in $R$. 
\end{thm}

The method underlying this complexity result is based on
Algorithms~\ref{alg:baset},~\ref{alg:baset-bis} and~\ref{alg:modexp}, which
are presented in Section~\ref{sec:msb}. Along the way, using the MSB-first
Algorithm~\ref{alg:baset}, we improve the cost of \emph{polynomial modular
exponentiation}, which is a central algorithmic task in computer algebra, with
many applications. Since this result has an interest \emph{per se}, we isolate
it here as our last complexity result.

\begin{thm} \label{thm:main-modexp}  
Given $N\in\N$ and a polynomial $\Gamma(x)$ in $R[x]$ of degree~$d$, one can compute $x^N \bmod \Gamma(x)$   using
\begin{equation*} 
	2 \, \M(d)  \lceil\log (N+1)\rceil + \M(d)
\end{equation*} 
arithmetic operations in $R$. 
\end{thm}     

This cost of $\sim 2\, \M(d) \log N$ compares favorably with the currently
best estimate of $\sim 3\, \M(d) \log N$ obtained by square-and-multiply in
the quotient ring $R[x]/(\Gamma(x))$, combined with fast modular
multiplications performed either classically~\cite[Corollary~9.7]{GaGe13}, or
using Montgomery's algorithm~\cite{Montgomery85}. In the FFT setting, the gain
is even larger, and our results improve on the best estimates, due to
Mih\u{a}ilescu~\cite{Mihailescu08}.


\subsection{Structure of the paper}

In Section~\ref{sec:lsb} we propose our LSB-first (least significant bit first) algorithm for computing the $N$-th term of a C-recursive sequence.
We design in Section~\ref{sec:msb} a second algorithm, 
which is an MSB-first (most significant bit first) variant, 
and discuss several algorithmic applications, including polynomial modular exponentiation,
powering of matrices and high-order lifting.
In Section~\ref{sec:FFT} we specialize and analyze Algorithm~\ref{alg:base} in the specific FFT
setting, where polynomial multiplication is based on Discrete Fourier
Transform techniques, and we compare it with the FFT-based Fiduccia's algorithm.  
We conclude in Section~\ref{sec:conclusion} by
a summary of results and plans of future work.
 

\section{The LSB-first algorithm and applications}  \label{sec:lsb}   

We will prove Theorem~\ref{thm:main} in \S\ref{sec:proof}, where we propose
the first main algorithms (Algorithms~\ref{alg:base} and~\ref{alg:base-rec}),
which are faster than Fiduccia's algorithm.
Then, in~\S\ref{sec:fibo} we instantiate them in the particular
case of the Fibonacci sequence. The resulting algorithm is
competitive with state-of-the-art algorithms.

\subsection{The LSB-first algorithm: Proof of Theorem~\ref{thm:main}}   \label{sec:proof}

The algorithm underlying Theorem~\ref{thm:main} is very natural.
Let us sketch now its main idea.

First, it is classical~\cite{Ranum11} that the generating functions of
linearly recurrent sequences are \emph{rational}.
As a consequence, computing terms of a
linearly recurrent sequence is equivalent to computing coefficients in the power
series expansion (at the origin) of a rational function.         
More precisely, let us attach to the recurrence~\eqref{eq:rec} the polynomial
$Q(x):=1 - c_{d-1} x - \cdots - c_0 x^d$, that is the reversal of the
characteristic polynomial $\Gamma(x)= x^d-\sum_{i=0}^{d-1} c_i x^i$ of recurrence~\eqref{eq:rec}. 
Let us denote by~$F(x)$ the generating function of the sequence~$(u_n)_{n \geq 0}$,
	\[ F(x) := \sum_{n \geq 0} u_n x^n.\]
Then, there exists a polynomial $P(x)$ in $R[x]$ of degree less than $d$ such
that $F(x) = P(x)/Q(x)$ in $R[[x]]$. This is immediately seen by checking
that, for any $n \geq 0$, the coefficient of $x^{n+d}$ in the power series
$P(x) := Q(x) \cdot F(x)$ is equal to $ u_{n+d} - c_{d-1} u_{n+d-1} - \cdots - c_0
u_n$, hence it is zero by~\eqref{eq:rec}, and therefore the power series $P(x)$ is in fact a polynomial of degree less than~$d$. Moreover, the coefficients of $P(x)$ 
can be determined from the recurrence~\eqref{eq:rec} and from the initial terms $u_0, \ldots, u_{d-1}$ by using $\M(d)$ operations in~$R$.    

We are thus reduced to the question of determining the $N$-th
coefficient~$u_N$ of the rational power series $F(x) = P(x)/Q(x)$. Our new algorithm is
based on the following observation. The polynomial $Q(x) Q(-x)$ is even, so it writes $V(x^2)$ for some $V\in R[x]$ of degree $d$. Then,
denoting by $U(x)$ the polynomial $P(x) Q(-x)$, of degree less than $2d$, and
by $U_{\rm e}$ and $U_{\rm o}$ the even and the odd parts of $U$, that is
$U(x) = U_{\rm e}(x^2) + x \cdot U_{\rm o}(x^2)$, we have          
\[
\frac{P(x)}{Q(x)} 
= \frac{P(x) Q(-x)}{Q(x) Q(-x)} 
=  \frac{U_{\rm e}(x^2)}{V(x^2)} + x \cdot \frac{U_{\rm o}(x^2)}{V(x^2)}, 
\]                                                                        
which implies that the $N$-th coefficient in the series expansion of $P/Q$ is 
\[
 [x^N]\; \frac{P(x)}{Q(x)} = 
\begin{cases}
     [x^{\frac{N}{2}}] \; \frac{U_{\rm e}(x)}{V(x)}   ,&\text{if $N$ is even,}\\
[x^{\frac{N-1}{2}}] \; \frac{U_{\rm o}(x)}{V(x)} ,&\text{else.}\end{cases}
\] 
In other words, the problem of computing the $N$-th term of a rational
function $P/Q$ of degree $d$ is reduced to that of computing the term of
index~$\lfloor N/2 \rfloor$ of another rational function of degree $d$, that
can be deduced from $P/Q$ by using two polynomial multiplications in degree
$d$. The desired coefficient $u_N$ is computed after repeating this process at
most $\lceil\log (N+1)\rceil$ times, and the complexity estimate in
Theorem~\ref{thm:main} is easily deduced.

\smallskip

Notice that, by the duality between linearly recurrent sequences and rational
functions, the new algorithm admits a nice and simple interpretation directly
at the level of recurrences. The sequence $(u_n)_{n \geq 0}$ is determined by
the recurrence~\eqref{eq:rec} (encoded by the denominator $Q(x)$) and by the
initial conditions $u_0, \ldots, u_{d-1}$ (encoded by the numerator $P(x)$).
To compute the $N$-th coefficient $u_N$, the new method builds a different
recurrence (encoded by $V(x) = Q(\sqrt{x}) Q(-\sqrt{x})$) still with constant
coefficients and of the same order~$d$, together with new initial conditions
(encoded by $P(x) Q(-x)$)\footnote{In fact, the new sequence consists of the
even (or odd) terms of the original sequence.}.
Computing $u_N$ is thus reduced to computing the term of index~$\lfloor N/2
\rfloor$ of the new sequence; and this reduction is applied  at
most $\lceil\log (N+1)\rceil$ times.

\begin{algorithm}[t]
{\bf Assumptions:} $Q(0)$ invertible and $\deg(P) < \deg(Q) =: d$
\begin{algorithmic}[1]    
\While{$N \ge 1$}
\State $U(x) \gets P(x)Q(-x)$    
\Comment{$U=\sum_{i=0}^{2d-1} U_i x^i$}
\If{$N$ is even}
\State $P(x) \gets \sum_{i=0}^{d-1} U_{2i} x^i$
\Else
\State $P(x) \gets \sum_{i=0}^{d-1} U_{2i+1} x^i$
\EndIf
\State $A(x) \gets Q(x)Q(-x)$
\Comment{$A=\sum_{i=0}^{2d} A_i x^i$}
\State $Q(x) \gets \sum_{i=0}^{d} A_{2i} x^i$
\State $N \gets \lfloor N/2\rfloor$
\EndWhile
\State \Return $P(0)/Q(0)$
\end{algorithmic}
\caption{(\textcolor{red}{\textsf{OneCoeff}}) \quad {\bf Input}: $P(x)$, $Q(x)$, $N$\qquad {\bf Output}: $[x^N]\, \frac{P(x)}{Q(x)}$}
\label{alg:base}
\end{algorithm}

\begin{algorithm}[t]
{\bf Assumptions:} 
$\Gamma(x) = x^d-\sum_{i=0}^{d-1} c_i x^i$ with  $c_0 \neq 0$
\begin{algorithmic}[1]    
\State $Q(x) \gets x^d \Gamma(1/x)$    
\State $P(x) \gets (u_0 + \cdots + u_{d-1} x^{d-1}) \cdot Q(x) \bmod \, x^d$    
\State \Return $[x^N] P(x)/Q(x)$
\Comment using Algorithm~\ref{alg:base}    
\end{algorithmic}
\caption{(\textcolor{red}{\textsf{OneTerm}}) \quad {\bf Input}: rec.~\eqref{eq:rec}, $u_0,\ldots,u_{d-1}$, $N$\quad {\bf Output}: $u_N$}
\label{alg:base-rec}
\end{algorithm}            

\smallskip

The proposed algorithm for computing $[x^N] P(x)/Q(x)$ is summarized in
Algorithm~\ref{alg:base}, and its immediate consequence for computing the
$N$-th term of the linearly recurrent sequence $(u_n)_{n \geq 0}$ defined
by~eq.~\eqref{eq:rec} is displayed in Algorithm~\ref{alg:base-rec}.
Algorithm~\ref{alg:base} has complexity $2 \, \M(d) \lceil\log (N+1)\rceil$
and Algorithm~\ref{alg:base-rec} has complexity $2 \, \M(d) \lceil\log (N+1)\rceil + \M(d)$, which proves in Theorem~\ref{thm:main}.

\smallskip

Note that Algorithms~\ref{alg:base} and~\ref{alg:base-rec} use an idea similar to the ones in~\cite{BoChDu16,BoCaChDu19} which were dedicated to the larger class of {algebraic} power series, but restricted to {positive characteristic} only.   
Algorithm~\ref{alg:base} also shares common features with the technique of \emph{section operators}~\cite[Lemma~4.1]{AlSh92} used by Allouche and Shallit to compute the $N$-th term of $k$-regular sequences~\cite[Corollary~4.5]{AlSh92} in $O(\log N)$ ring operations.

Algorithm~\ref{alg:base} can be interpreted at the level of recurrences as computing
$\sim \log N$ new recurrences produced by the Graeffe process, which is a
classical technique to compute the largest root of a real polynomial~\cite{Householder59,Pan87,Pan97}.    
Interestingly, the Graeffe process has been used in a purely algebraic context 
by Sch\"{o}nhage in~\cite[\S 3]{Schonhage00} for computing the reciprocal of a power series, see also~\cite[\S 2]{CDJPS07}.
However, our paper seems to be the first reference where the Graeffe process
and the section operators approach are combined together.   

\subsection{New algorithm for Fibonacci numbers} \label{sec:fibo}

\begin{algorithm}[ht]
{\bf Assumptions:} $N \geq 2$
\begin{algorithmic}[1]
\State $c \gets 3$    
\If{$N$ is even}
\State $[a,b] \gets [0,1]$
\Else
\State $[a,b] \gets [1,-1]$   
\EndIf
\State $N \gets \lfloor N/2\rfloor$
\While{$N > 1$}
\If{$N$ is even}
\State $b \gets a+b\cdot c$
\Else
\State $a \gets b+a\cdot c$ 
\EndIf
\State $c \gets c^2-2$
\State $N \gets \lfloor N/2\rfloor$
\EndWhile                   \State \Return $b+a \cdot c$
\end{algorithmic}
\caption{(\textcolor{red}{\textsf{NewFibonacci}}) \quad {\bf Input}: $N$\qquad {\bf Output}: $F_N$}
\label{alg:newfibo}
\end{algorithm}   

\begin{algorithm}[ht]
{\bf Assumptions:} $N \geq 2$ and $N$ is a power of $2$
\begin{algorithmic}[1]
\State $[b,c] \gets [1,3]$    
\State $N \gets \lfloor N/2\rfloor$
\While{$N > 2$}
\State $b \gets b\cdot c$
\State $c \gets c^2-2$
\State $N \gets \lfloor N/2\rfloor$
\EndWhile                   \State \Return $b \cdot c$
\end{algorithmic}
\caption{\quad  {\bf Input}: $N$\qquad {\bf Output}: $F_N$}
\label{alg:newfibo-powerof2}
\end{algorithm}

To illustrate the mechanism of Algorithm~\ref{alg:base}, let us instantiate it
in the particular case of the Fibonacci sequence defined by
\[ F_0 = 0, F_1 = 1, \quad F_{n+2} = F_{n+1} + F_n, \; n \geq 0.\]
The generating function $\sum_{n \geq 0 } F_n x^n$ equals $x/(1-x-x^2)$.     
Therefore, the coefficient $F_N$ is equal to 
\[
[x^{N}]\; \frac{x}{1-x-x^2} 
= [x^{N}]\; \frac{x(1+x-x^2)}{1-3x^2+x^4}
= 
\begin{cases}
     [x^{\frac{N}{2}}] \; \frac{x}{1-3x+x^2}   ,&\text{if $N$ is even,}\\
[x^{\frac{N-1}{2}}] \; \frac{1-x}{1-3x+x^2} ,&\text{else.}\end{cases}    
\]
The computation of $F_N$ is reduced to that of a coefficient of the form
\[
[x^{N}]\; \frac{a+bx}{1-cx+x^2} 
= [x^{N}]\; \frac{(a+bx)(1+cx+x^2)}{1-(c^2-2)x^2+x^4}
= 
\begin{cases}
     [x^{\frac{N}{2}}] \; \frac{a+(bc+a)x}{1-(c^2-2)x+x^2}   ,&\text{if $N$ is even,}\\
[x^{\frac{N-1}{2}}] \; \frac{(ac+b)+bx}{1-(c^2-2)x+x^2} ,&\text{else.}\end{cases}    
\] 
This yields Algorithm~\ref{alg:newfibo} for the computation of $F_N$\footnote{Notice that the same algorithm can be used to compute efficiently the 
$N$-th Fibonacci polynomial, or the $N$-th Chebychev polynomial.
Fibonacci polynomials in $R[t]$ are defined by $F_{n+2}(t) = t \cdot F_{n+1}(t) + F_n(t)$ with $F_0(t)=1$ and $F_1(t)=1$. It is sufficient to initialize
$c$ to $t^2+2$ (instead of 3) and $b$ to $t$ when $N$ is even (instead of $0$). The complexity of this algorithm is $O(\M(N))$ operations in~$R$, which is quasi-optimal.}.     
A close inspection reveals that this algorithm computes $F_N$ by a recursive use of the formula
\[
F_N =  L_{2^{\lfloor \log N \rfloor}} \cdot F_{N - 2^{\lfloor \log N \rfloor}} + (-1)^N \cdot F_{2^{1 + \lfloor \log N \rfloor}-N}, 
\]                                                  
which is a particular instance of the classical formula
       \[F_{n+m} =  L_m  F_n  + (-1)^n F_{m-n} \]
relating the Fibonacci numbers and the Lucas numbers $L_n = F_{n+1} + F_{n-1}$.

\smallskip When $N$ is a power of 2, then Algorithm~\ref{alg:newfibo} 
degenerates into Algorithm~\ref{alg:newfibo-powerof2}.  
This is equivalent to Algorithm \texttt{fib}$(n)$ in~\cite[Fig.~6]{CuHa89}\footnote{This algorithm had also appeared before, in Knuth's book~\cite[p.~552, second solution]{Knuth69}.}. 
It uses $2 \log(N) - 3$ products (of which $\log(N) - 2$ are squarings)  and $\log(N) - 2$ subtractions.

When $N$ is arbitrary, Algorithm~\ref{alg:newfibo} has essentially the same cost: it uses at most  
$2 \log(N) - 1$ products (of which at most $\log(N) - 1$ are squarings) 
and 
$2 \lfloor \log(N) \rfloor - 1$ additions/subtractions.
In contrast,  \cite[Fig.~6]{CuHa89} uses a more complex algorithm,
with higher cost. 
An example of execution of our Algorithm~\ref{alg:newfibo} for computing $F_{43}$ is 
explicitly displayed in Table~\ref{tab:F43-bis}.

A nice feature of Algorithm~\ref{alg:newfibo} is not only that it is simple
and natural, but also that its arithmetic and bit complexity matches the
complexity of the state-of-art algorithms for computing Fibonacci
numbers~\cite{Takahashi00}.

      \begin{table}[hb]
  \centering
\begin{tabular}{|c|c|c|c|}
 $N$ & $a$ & $b$ & $c$\\
  \hline
  $21$ & $1$ & $-1$ & $3$ \\
  $10$ & $ 1 \times 3 -1 = 2$ &  & $3^2-2=7$\\
  $5$ & & $(-1)\times 7 + 2 = -5$ & $7^2-2=47$\\
  $2$ & $2 \times 47 - 5= 89 $ & & $47^2-2=2207$\\
  $1$ &  & $(-5) \times 2207 + 89 $ & $2207^2-2=4870847$\\
	& &= $-10946$  & \\
  $0$ & $89 \times 4870847 -10946$  &  & \\
     & $=  433494437 $  & & \\
\end{tabular}
\caption{Computation of $F_{43} = 433494437$ using the new algorithm.}  
\label{tab:F43-bis}
\end{table}


\section{The MSB-first algorithm and applications}  \label{sec:msb}   

We present in \S\ref{ssec:msb-first} a ``most significant bit'' (MSB) variant
(Algoritm~\ref{alg:baset}) of Algoritm~\ref{alg:base}. Then we discuss various
applications of Algorithms~\ref{alg:base} and~\ref{alg:baset}.
In~\S\ref{sec:modexp} we design a faster algorithm for polynomial modular
exponentiation, that we use in~\S\ref{sec:newFiduccia} to design a faster
Fiduccia-like algorithm for computing a slice of $d$ terms of indices $N-d+1,
\ldots, N$ in $\sim 2 \, \M(d) \log N$ operations.


\subsection{The MSB-first algorithm}           \label{ssec:msb-first}

In Fiduccia's algorithm (\S\ref{sec:fiduccia}), the $N$-th coefficient $u_N$
in the power series expansion $\sum_{i\geq 0} u_i x^i$ of $P/Q$ is given by
the inner product $\langle x^N\bmod \Gamma (x), v_0\rangle$, where $\Gamma$ is
the reversal polynomial of~$Q$ and $v_0$ is the vector of initial coefficients
$\begin{bmatrix} u_{0} & \cdots & u_{d-1} \end{bmatrix}$. Here, $x^N \bmod
\Gamma (x)$ depends only on the linear recurrence equation~\eqref{eq:rec}, and
is independent of the initial terms~$v_0$. Hence, if we want to compute the
$N$-th terms of $k$ different linearly recurrent sequences that share the same
linear recurrence equation~\eqref{eq:rec}, we can first determine $\rho(x):=
x^N\bmod \Gamma (x)$, and then $\langle \rho, v_0^{(i)}\rangle$ for
$i=1,\dotsc,k$, where~$v_0^{(i)}$ denotes the vector of $d$ initial terms of
the $i$-th sequence. The total arithmetic complexity of this algorithm is
$O(\M(d)\log N + kd)$; this is faster than Fiduccia's algorithm repeated
independently $k$ times, with cost $O(k \, \M(d)\log N)$.

On the other hand, in Algorithm~\ref{alg:base}, we iteratively update both the
denominator and the numerator, and each new numerator depends on the original
numerator $P(x)$ which encodes the initial $d$ terms of the sequence. Hence,
it is not a priori clear how to obtain with Algorithm~\ref{alg:base} the good
feature of Fiduccia's algorithm mentioned above.

In this section, we present an algorithm that computes $u_N$ with arithmetic
complexity equal to that of Algorithm~\ref{alg:base} and which, in addition,
achieves the cost $O(\M(d)\log N + kd)$ for the above problem with $k$
sequences.

While Algorithm~\ref{alg:base} looks at $N$ from the least significant bit
(LSB), the main algorithm presented in this section (Algoritm~\ref{alg:baset})
looks at $N$ from the most significant bit (MSB). In fact,
Algoritm~\ref{alg:baset} is in essentially equivalent to ``the transposition''
of Algorithm~\ref{alg:base} in the sense of~\cite{BoLeSc03}. This is the
reason why this MSB-first algorithm has exactly the same complexity as
Algorithm~\ref{alg:base}. However, in order to keep the presentation
self-contained, we are not going to appeal here to the general machinery of
algorithmic transposition tools, but rather derive the transposed algorithm
``by hand'', using a direct reasoning.

\begin{algorithm}[t]
{\bf Assumptions:} $Q(0)$ invertible and $\deg(Q) =: d$
\begin{algorithmic}[1]
\Function{$\mathcal{F}$}{$N$, $Q(x)$}
\If{$N=0$}
\State \Return $x^{d-1}/Q(0)$
\EndIf
\State $A(x) \gets Q(x)Q(-x)$
\Comment{$A=\sum_{i=0}^{2d} A_i x^i$}
\State $V(x) \gets \sum_{i=0}^d A_{2i} x^i$
\State $W(x) \gets \mathcal{F}(\lfloor N/2\rfloor, V(x))$
\If{$N$ is even}
\State $S(x) \gets x W(x^2)$
\Else
\State $S(x) \gets W(x^2)$
\EndIf
\State $B(x) \gets Q(-x) S(x)$
\Comment{$B=\sum_{i=0}^{3d-1} B_i x^i$}
\State \Return $\sum_{i=0}^{d-1} B_{d+i} x^i$
\EndFunction
\end{algorithmic}
\caption{(\textcolor{red}{\textsf{SliceCoeff}}) \quad {\bf Input}: $Q(x)$, $N$\qquad {\bf Output}: $\mathcal{F}_{N,d}(1/Q(x))$}
\label{alg:baset}
\end{algorithm}

In order to compute the coefficient $[x^N]\, P(x)/Q(x)$, it is sufficient to
compute the $(N-d+1)$-th term to the $N$-th term of $1/Q(x)$ since the degree
of $P(x)$ is at most $d-1$.
Let $\mathcal{F}_{N,d}(\sum_{i\ge 0} a_i x^i) := \sum_{i=0}^{d-1} a_{N-d+1+i} x^i$. Our goal is to compute $\mathcal{F}_{N,d}(1/Q(x))$.
We have the sequence of equalities
\begin{align*}
    \mathcal{F}_{N,d}\left(\frac1{Q(x)}\right) &= \mathcal{F}_{N,d}\left(\frac{Q(-x)}{Q(x)Q(-x)}\right)\\
    &= \mathcal{F}_{N,d}\left(Q(-x)x^{N-2d+1}\mathcal{F}_{N,2d}\left(\frac1{Q(x)Q(-x)}\right)\right)\\
    &= \mathcal{F}_{2d-1,d}\left(Q(-x)\mathcal{F}_{N,2d}\left(\frac1{V(x^2)}\right)\right),
\end{align*}
where $V(x^2) := Q(x)Q(-x)$. 

In the second equality, we ignore the terms of $1/V(x^2)$ except for the
$(N-2d+1)$-th term to the $N$-th term. In the third equality, we use the fact
that $\mathcal{F}_{N,d}(xA(x)) = \mathcal{F}_{N-1,d}(A(x))$.
Let now $W(x) := \mathcal{F}_{\lfloor N/2\rfloor,d}(1/V(x))$.
Then, it is easy to see that
\begin{align*}
    \mathcal{F}_{N,d}\left(\frac1{Q(x)}\right) &= \mathcal{F}_{2d-1,d}\left(Q(-x)S(x)\right),
\end{align*}
where
\begin{align*}
    S(x) &:=
    \begin{cases}
    x W(x^2),& \text{if $N$ is even}\\
    W(x^2),& \text{else.}
    \end{cases}
\end{align*}
The resulting method for computing $\mathcal{F}_{N,d}(1/Q(x))$ is summarized
in Algorithm~\ref{alg:baset}, and its immediate applications to the computation
of $[x^N] P/Q$, and to  $[x^N] P^{(i)}/Q$ for several $i=1,\ldots,k$, are displayed in Algorithms~\ref{alg:baset-bis} and~\ref{alg:baset-ter}.

\begin{algorithm}[t]
{\bf Assumptions:} $Q(0)$ invertible and $\deg(P) < \deg(Q) =: d$
\begin{algorithmic}[1]    
\State $U\gets \mathcal{F}_{N,d}(1/Q(x))$ using Algorithm~\ref{alg:baset}   
\Comment{$U = u_{N-d+1} + \cdots + u_N x^{d-1}$}  
\State \Return $p_0 u_N + \cdots + p_{d-1} u_{N-d+1}$
\Comment{$P=\sum_{i=0}^{d-1} p_i x^i$}  
\end{algorithmic}
\caption{(\textcolor{red}{\textsf{OneCoeffT}}) \quad {\bf Input}: $P(x)$, $Q(x)$, $N$\qquad {\bf Output}: $[x^N]\, \frac{P(x)}{Q(x)}$}
\label{alg:baset-bis}
\end{algorithm}

\begin{algorithm}[t]
{\bf Assumptions:} $Q(0)$ invertible and $\deg(P) < \deg(Q) =: d$
\begin{algorithmic}[1]    
\State $U\gets \mathcal{F}_{N,d}(1/Q(x))$ using Algorithm~\ref{alg:baset}   
\Comment{$U = u_{N-d+1} + \cdots + u_N x^{d-1}$}  
\State \Return $p_0^{(j)} u_N + \cdots + p_{d-1}^{(j)} u_{N-d+1}, j=1,\ldots,k$
\Comment{$P_j=\sum_{i=0}^{d-1} p_i^{(j)} x^i$}  
\end{algorithmic}
\caption{\quad {\bf Input}: $P_1, \ldots, P_k$, $Q$, $N$\quad {\bf Output}: $[x^N]\, \frac{P_j}{Q}, j=1,\ldots,k$}
\label{alg:baset-ter}
\end{algorithm}    

Let us analyze the complexity of Algorithms~\ref{alg:baset}
and~\ref{alg:baset-bis} more carefully. At each step,
Algorithm~\ref{alg:baset} computes $Q(x)Q(-x)$ and $Q(-x)S(x)$, where the
degrees of $Q(x)$ and $S(x)$ are $d$ and at most $2d-1$, respectively. Hence a
direct analysis concludes that its complexity is $3 \, \M(d) \log N$
operations in $R$. However, an improvement comes from the remark that not all
coefficients of $Q(-x)S(x)$ are needed: it is sufficient to compute the $d$-th
coefficient to the $(2d-1)$-th coefficient of $Q(-x)S(x)$. This operation is
known as ``the middle product'', and can be performed with exactly the same
arithmetic complexity as the standard product of two polynomials of degrees
$d$ and~$d-1$~\cite{HaQuZi04,BoLeSc03}. Therefore, if steps~11 and~12 of
Algorithm~\ref{alg:baset} are performed ``at once'' using a middle product,
then the arithmetic complexity drops to $2 \, \M(d) \log N$. This
complexity is also inherited by Algorithm~\ref{alg:baset-bis}, which uses at
most $2d$ additional operations in the last step.

It should be obvious at this point that the slight variant
Algorithm~\ref{alg:baset-ter} of Algorithm~\ref{alg:baset-bis} achieves
arithmetic complexity $O(\M(d)\log N + kd)$ for the aforementioned problem
with $k$ sequences, and more precisely its cost is of at most $(2 \, \M(d) +
d)\log N + 2kd$ operations in~$R$.

In conclusion, Algorithm~\ref{alg:baset} achieves the same arithmetic
complexity as Algorithm~\ref{alg:base}
and it extends to Algorithms~\ref{alg:baset-bis} and~\ref{alg:baset-ter}. All
algorithmic techniques specific to the FFT setting, that we will describe in
Section~\ref{sec:FFT}, can also be applied to Algorithms~\ref{alg:baset},
\ref{alg:baset-bis} and \ref{alg:baset-ter}, yielding the same
complexity gains.


\subsection{Faster modular exponentiation} \label{sec:modexp}

The algorithms of \S\ref{ssec:msb-first} are not only well-suited to compute
the $N$-th terms of several sequences satisfying the same recurrence relation.
In this section, we show that they also permit a surprising application to the
computation of \emph{polynomial modular exponentiations}. This fact has many
consequences, since modular exponentiation is a central algorithmic task in
algebraic computations. In~\S\ref{sec:newFiduccia}, we will discuss a first
application in relation with the main topic of our article. Namely, we will
design a new Fiduccia-style algorithm for the computation of the $N$-th term,
and actually of a whole slice of $k \geq d$ terms, in $2 \, \M(d) \log N +
O((k+d) \M(d)/d)$ arithmetic operations. More consequences will be separately
discussed in~\S\ref{sec:applications}.

\smallskip
Assume we are given a polynomial $\Gamma(x) \in R[x]$ of degree~$d$, an
integer~$N$, and that we want to compute $\rho(x):=x^N \bmod \Gamma(x)$.
Without loss of generality, we may assume $\Gamma(0) \neq 0$.
Let $Q(x) \in R[x]$ be the reversal of $\Gamma(x)$, that is $Q(x):=x^d \Gamma(1/x)$.
Let us denote the power series expansion of $1/Q$ by $\sum_{i \geq 0} a_i x^i$.   Then, equation~\eqref{eq:fiduccia} implies that   
\begin{equation} \label{sys:hankel}
\begin{bmatrix}
	a_{N} & \cdots & a_{N+d-1}
\end{bmatrix}      
=
{\bf r}
\times {\bf H},
\end{equation} 
where 
$
{\bf r} 
=
\begin{bmatrix}
	r_{0} & \cdots & r_{d-1}
\end{bmatrix}    
$
with $\rho=\sum_{i=0}^{d-1} r_i x^i$ 
and ${\bf H}$ is the Hankel matrix
 \[
{\bf H}
:=
\begin{bmatrix}
	a_{0} & \cdots & a_{d-1} \\
	a_{1} & \cdots & a_{d} \\
	       &   \vdots    &   \\
	a_{d-1} & \cdots & a_{2d-2} \\
\end{bmatrix}.  
\]  
Note that the matrix ${\bf H}$ is invertible, as its determinant is 
equal (up to a sign) to $([x^d] Q)^{d-1} = \Gamma(0)^{d-1}$.          
Therefore, ${\bf r}$ (and thus $\rho$) can be found by  

\begin{enumerate}
	\item[(1)] computing 
$ \begin{bmatrix}
	u_{N} & \cdots & u_{N+d-1}
\end{bmatrix}$ 
using Algorithm~\ref{alg:baset};

	\item[(2)] solving the Hankel linear system~\eqref{sys:hankel}.
\end{enumerate}
The arithmetic complexity of step (1) is $2 \, \M(d) \log (N)  + O(d \log N)$,
while step (2) has negligible cost $O(\M(d) \log d)$ using~\cite{BGY80}, see also~\cite[Ch.~2, \S5]{BiPa94}.    

\begin{algorithm}[t]
{\bf Assumptions:} ${\rm lc}(\Gamma)$ invertible, $\Gamma(0) \neq 0$ and $\deg(\Gamma) =: d$
\begin{algorithmic}[1]    
\State $Q(x) \gets x^d \Gamma(1/x)$    
\State $u(x) \gets \mathcal{F}_{N,d}(1/Q(x))$
\Comment using Algorithm~\ref{alg:baset}
\State $v(x) \gets u(x) Q(x) \bmod x^d$
\State \Return $v(1/x) x^{d-1}$
\end{algorithmic}
\caption{(\textcolor{red}{\textsf{NewModExp}}) \quad {\bf Input}: $\Gamma(x)$, $N$\qquad {\bf Output}: $x^N \bmod \Gamma(x)$}
\label{alg:modexp}
\end{algorithm}     

It is actually possible to improve a bit more on this algorithm, by using the
next lemma. 

\begin{lem}  \label{lem:power}
Let $N \in \N$ and let $\Gamma(x) \in R[x]$ be of degree~$d$ with $\Gamma(0) \neq 0$.
Let $Q(x) \in R[x]$ be its reversal, $Q(x):=x^d \Gamma(1/x)$.   
Denote its reciprocal $1/Q$ by $\sum_{i \geq 0} a_i x^i$, and let
$u(x)$ be $\mathcal{F}_{N,d}(1/Q(x))
= a_{N-d+1} + \cdots + a_Nx^{d-1}$.  
Define $v(x)$ to be $u(x) Q(x) \bmod x^d$.
Then  $x^N \bmod \Gamma(x) = v(1/x) x^{d-1}$.
\end{lem}  

\begin{proof}
Write the Euclidean division $x^N = L(x) \cdot \Gamma(x) + \rho(x)$,
where $\deg(L) = N-d$ and $\rho(x) = r_0 + \cdots + r_{d-1} x^{d-1}$.
Replacing $x$ by $1/x$ on both sides, and then multiplying by $x^N$ yields
$1 = L_{\rm rev}(x) \cdot Q(x) + x^{N-d+1} \cdot \tilde{\rho}(x)$,
where $L_{\rm rev}(x) = x^{N-d} \cdot L(1/x)$ and  $\tilde{\rho}(x) = x^{d-1} 
\cdot \rho(1/x)$.
In other words
\[ \frac{1}{Q(x)} = L_{\rm rev}(x) +  x^{N-d+1} \cdot \frac{\tilde{\rho}(x)}{Q(x)}.\]
Since $L_{\rm rev}(x)$ has degree at most $N-d$, it follows that 
$u(x) = \frac{\tilde{\rho}(x)}{Q(x)} \bmod x^d$. Therefore, $\tilde{\rho}(x)$ is equal to 
$v(x)$, and the conclusion follows.
\end{proof}	  

The merit of Lemma~\ref{lem:power} is that it shows that computing $x^N \bmod
\Gamma(x)$ can be reduced to computing $\mathcal{F}_{N,d}(1/Q(x))$, plus a few
additional operations with negligible cost $\M(d)$. The resulting method is
presented in Algorithm~\ref{alg:modexp}, whose complexity is $2 \,\M(d) \log N
+ \M(d)$. This proves Theorem~\ref{thm:main-modexp}.

Notice that Algorithm~\ref{alg:modexp} is simpler, and faster by a factor of
1.5, than the classical algorithm based on binary powering in the quotient
ring $R[x]/(\Gamma(x))$. Algorithm~\ref{alg:modexp} admits a specialization
into the FFT setting, with complexity $\sim \frac23 \, \M(d) \log N$, in the
spirit of \S\ref{sec:FFT} below. Similarly to the case of
Algorithm~\ref{alg:FFT} in~\S\ref{sec:FFT}, the FFT variant of
Algorithm~\ref{alg:modexp} is faster by a factor of~$2.5$ than Shoup's
(comparatively simple) algorithm~\cite[\S7.3]{Shoup95}, and by a factor
of~$1.625$ than the (much more complex) algorithm of
Mih\u{a}ilescu~\cite{Mihailescu08}.

This speed-up might be beneficial for instance in applications to polynomial
factoring in $\F_p[x]$, where one time-consuming step to factor~$f \in
\F_p[x]$ is the computation of $x^p \bmod f$, see~\cite[Algorithms~14.3, 14.8,
14.13, 14.15, 14.31, 14.33 and 14.36]{GaGe13}, and
also~\cite{Shoup95,Lecerf10}.      

It might also be so in point-counting methods such as Schoof's algorithm and
the Schoof-Elkies-Atkin (SEA) algorithm~\cite[Ch.~VII]{BSS99}, the second one
being the best known method for counting the number of points of elliptic
curves defined over finite fields of large characteristic. Indeed, the bulks
of these algorithms are computations of $x^q$ modulo the ``division
polynomial''~$f_\ell(x)$ and of $x^q$ modulo the ``modular
polynomial''~$\Phi_\ell(x)$, where $\ell= O(\log(q))$ and $\deg(f_\ell) =
O(\ell^2)$, $\deg(\Phi_\ell) = O(\ell)$.

As a final remark, note that while Fiduccia's algorithm shows that computing
the terms of indices $N, \ldots, N+d-1$ of a linearly recurrent sequence of
order $d$ can be reduced to polynomial modular exponentiation ($x^N \bmod
\Gamma(x)$), Algorithm~\ref{alg:modexp} shows that the converse is also true:
polynomial modular exponentiation can be reduced to computing the terms of
indices $N, \ldots, N+d-1$ of a linearly recurrent sequence of order $d$.
Therefore, \emph{these two problems are computationally equivalent}. To our
knowledge, this important fact seems not to have been noticed before.

\subsection{A new Fiduccia-style algorithm} \label{sec:newFiduccia}

We conclude this section by discussing a straightforward application of Algorithm~\ref{alg:modexp}. This is based on the next equality,
generalizing~\eqref{sys:hankel} to any $k \geq 1$:   
\begin{equation} \label{sys:hankel-long}
\begin{bmatrix}
	u_{N} & \cdots & u_{N+k-1}
\end{bmatrix}      
=
{\bf r}
\times {\bf H}_k,
\end{equation} 
where as before
$
{\bf r} 
=
\begin{bmatrix}
	r_{0} & \cdots & r_{d-1}
\end{bmatrix}    
$
is the coefficients vector of $\rho=\sum_{i=0}^{d-1} r_i x^i$, 
with $\rho = x^N \bmod \Gamma(x)$,
and ${\bf H}_k$ is the Hankel matrix
 \[
{\bf H}_k
:=
\begin{bmatrix}
	u_{0} & \cdots & u_{d-1} & \cdots & \cdots & u_{k-1}\\
	u_{1} & \cdots & u_{d} & \cdots & \cdots & u_k \\ 
	       &   \vdots    & \vdots & \vdots & \vdots  \\
	u_{d-1} & \cdots & u_{2d-2} \cdots & \cdots & \cdots & u_{k+d-2}\\
\end{bmatrix}.  
\]     

The matrix ${\bf H}_k$ is built upon the first terms of the sequence $(u_n)_{n
\geq 0}$ satisfying recurrence~\eqref{eq:rec} with characteristic polynomial
$\Gamma = x^d-\sum_{i=0}^{d-1} c_i x^i$, or equivalently, from the power
series expansion of the rational function $P/Q$ with $Q(x) = x^d \,
\Gamma(1/x)$.

Note that the entries of ${\bf H}_k$ can be computed either from $P$ and $Q$,
or from the recurrence~\eqref{eq:rec} together with the initial terms $u_0,
\ldots, u_{d-1}$, using $O((k+d) \, \M(d)/d)$ arithmetic operations, by the
algorithm in~\cite[Thm.~3.1]{Shoup91}, see also~\cite[\S5]{BoLeSc03}. To
compute $u_N, \ldots, u_{N+k-1}$ it thus only remains to perform the
vector-matrix product~\eqref{sys:hankel-long}.

When $k=1$, the product ${\bf r} \times {\bf H}_1$ costs $2d$
operations and it yields the term~$u_N$.

\begin{algorithm}[t]
{\bf Assumptions:} 
$\Gamma(x) = x^d-\sum_{i=0}^{d-1} c_i x^i$ with  $c_0 \neq 0$
\begin{algorithmic}[1]    
\State $\rho(x) \gets x^N \bmod \Gamma(x)$    
\Comment using Algorithm~\ref{alg:modexp}    
\State $U(x) \gets u_0 + \cdots + u_{2d-2} x^{2d-2}$    
\Comment using Algorithm in~\cite[p.~18]{Shoup91}   
\State $V(x) \gets U(x) \cdot (x^d \cdot \rho(1/x))$    
\Comment{$V=\sum_{i=0}^{d-1} v_i x^i$}
\State \Return $[v_d, \ldots, v_{2d-1}]$
\end{algorithmic}
\caption{
\; {\bf Input}: rec.~\eqref{eq:rec}, $u_0,\ldots,u_{d-1}$, $N$\quad {\bf Output}: $u_N, \ldots, u_{N+d-1}$}
\label{alg:NewFiduccia}
\end{algorithm}    

When $k \geq d$, the product ${\bf r} \times {\bf H}_k$ can be reduced to the
polynomial multiplication of $r_{d-1} + \cdots + r_0 x^{d-1}$ by
$\sum_{i=0}^{k+d-2} u_i x^i$, and this can be performed using $\lceil
\frac{k+d}{d}\rceil \, \M(d)$ arithmetic operations. As a consequence, the
whole slice of coefficients $u_{N+i} = [x^{N+i}] P/Q$ for $i=0,\ldots,k-1$,  can be computed using Algorithm~\ref{alg:modexp} and
eq.~\eqref{sys:hankel-long} for a total cost of arithmetic operations equal to 
	\[2 \, \M(d)  \log N +
O\left(\frac{k+d}{d} \, \M(d) \right).\]
When $k=d$, this proves Theorem~\ref{thm:main-t}.
The corresponding algorithm is given as Algorithm~\ref{alg:NewFiduccia}.

We emphasize that this variant of Fiduccia's algorithm is different from Algorithm~\ref{alg:base}. It is actually a bit slower than Algorithm~\ref{alg:base}
when $k=1$. 
 However,  when $k>1$ terms are to be computed, it
should be preferred to repeating $k$ times Algorithm~\ref{alg:base}. 
It also compares favorably with Fiduccia's original algorithm, whose    
adaption to $k$ terms has arithmetic complexity
 \[ 3 \, \M(d) \log N + O\left( d \log(N) + \frac{k+d}{d} \, \M(d) \right).\]    


\subsection{Applications} \label{sec:applications}

In this section, we discuss three more applications of the MSB-first
algorithms (Algorithm~\ref{alg:baset} and~\ref{alg:modexp}) presented in
\S\ref{ssec:msb-first} and \S\ref{sec:modexp}. We deal with the case of
multiplicities (\S\ref{sec:mult}), and explain a new way to speed up
computations in that case. Then, we address other applications, to faster
powering of matrices (\S\ref{sec:powmat}) and to faster high-order lifting
(\S\ref{sec:HOL}). 

To simplify matters, we assume in this section that $R = \K$ is a field.

\subsubsection{The case with multiplicities} \label{sec:mult}

Hyun and his co-authors~\cite{HyMeSt19,HyMeScSt19} addressed the following
question: is it possible to compute faster the $N$-th term of a linearly
recurrent sequence when the characteristic polynomial of the recurrence has
multiple roots? By the Chinese Remainder Theorem, it is sufficient to focus on
the case where the characteristic polynomial is a pure power of a squarefree
polynomial. In other words, the main step of~\cite[Algorithm 1]{HyMeSt19} is
to compute $x^N \bmod Q$, where $Q=(Q^\star)^m$ and $Q^\star$ is the
squarefree part of $Q$. Under suitable invertibility conditions, the problem
is solved in~\cite{HyMeSt19,HyMeScSt19} in $O(\M(d^\star) \log N + \M(d) \log
d)$ operations in $\K$, where $d^\star = \deg (Q^\star)$ and $d=\deg(Q) = m \cdot d^\star$.
This cost is
obtained using an algorithm based on bivariate computations,  using the
isomorphisms between $\K[x]/(Q)$ and $\K[y,x]/(Q^\star(y),(x-y)^m)$ made
effective by the so-called \emph{tangling} / \emph{untangling} operations.
We now propose an alternatively fast, but simpler, algorithm with the
same cost.

Let us explain this on an example, for \textcolor{magenta}{\href{https://oeis.org/A037027}{``multiple-Fibonacci numbers''}}, that is
when $Q$ has the form $(Q^\star)^m$, with $Q^\star = 1-x-x^2$ and $d^\star=2, d=2 d^\star$. Assume we want
to compute the $N$-th coefficient $u_N$ in the power series expansion of
$(x/(1-x-x^2))^m$. The cost of Fiduccia's algorithm, and also of our new
algorithms, is $O(\M(m) \cdot \log N)$. Let us explain how we can lower this
to $O(\log N + \M(m) \log m)$.    
The starting point is the observation that, by the structure theorem of linearly recurrent sequences~\cite[\S A.(I)]{CeMiPi87} (see also~\cite[\S2]{Poorten89}), $u_N$ is of the form $u_m(N) \phi^N + v_m(N) \psi^N$, where $\phi$ and~$\psi$ are the two roots of $1+x=x^2$ and
$u_m, v_m$ are polynomials in $\overline{\K}[x]$ of degree less than~$m$. By an easy liner algebra argument, $u_N$ is thus equal to  $U_m(N) F_{N} + V_m(N) F_{N+1}$, where $U_m(x)$ and $V_m(x)$ are polynomials in $\K[x]$ of degree less than~$m$. These polynomials can be computed by (structured) linear algebra from the first $2m$ values of the sequence $(u_n)$, in complexity $O(\M(m) \log m)$.
For instance, when $d=2$, we have $U_2(x) = -(x+1)/5$ and $V_2(x) = 2 x/5$.
Once $U_m$ and $V_m$ are determined, it remains to compute $F_N$ and $F_{N+1}$ using Algorithm~\ref{alg:NewFiduccia} in $O(\log N)$ operations in~$\K$, then to return the value
$ U_m(N) \cdot F_N + V_m(N) \cdot F_{N+1}$.

The arguments extend to the general case and yields an algorithm of arithmetic complexity $2 \, \M(d^\star) \log N + O(\M(d) \log d)$.

\subsubsection{Faster powering of matrices} \label{sec:powmat}

Assume we are given a matrix $M \in \mathcal{M}_d(\K)$, an
integer~$N$, and that we want to compute the $N$-th power $M^N$ of $M$.                                                                

The arithmetic complexity of binary powering in $\mathcal{M}_d(\K)$ is
$O(d^\theta \log N)$ operations in~$\K$, where as before $\theta \in [2,3]$ is any
feasible exponent for matrix multiplication in $\mathcal{M}_d(\K)$.
A better algorithm consists in first computing the characteristic polynomial
$\Gamma(x)$ of the matrix~$M$, then the remainder $\rho(x):=x^N \bmod \Gamma(x)$,
and finally evaluating the polynomial $\rho(x)$ at~$M$. By the Cayley-Hamilton
theorem, $\rho(M) = M^N$. The most costly step is the computation of~$R$, which
can be done as explained in \S\ref{sec:modexp} using $\sim 2 \, \M(d) \log (N)$ operations in~$\K$. The cost of the other two
steps is independent of~$N$, and it is respectively $O(d^\theta \log d)$~\cite{Keller85} and
$O(d^{\theta +\frac12})$, this last cost being achieved using the
Paterson-Stockmeyer \emph{baby-step / giant-step} algorithm~\cite{PaSt73}.
The total cost of this algorithm is $2 \, \M(d) \log (N) + O(d^{\theta +\frac12})$.

Note that a faster variant (w.r.t.~$d$), of cost $2 \, \M(d) \log (N) +
O(d^{\theta} \log d)$, can be obtained
using~\cite[Corollary~7.4]{Giesbrecht95}. The corresponding algorithm is based
on the computation of the Frobenius (block-companion) form of the matrix~$M$,
followed by the powering of companion matrices, which again reduces to modular
exponentiation.                                                     

\subsubsection{Faster high-order lifting} \label{sec:HOL}   

The fastest known algorithms for polynomial linear algebra rely fundamentally
on an algorithmic technique introduced by
Storjohann~\cite{Storjohann02}, called \emph{high-order lifting}.

Given an invertible {polynomial matrix}~$A$ of degree~$d$, the problem is to compute the {high order components\/} $(C_{0},C_1), (C_2,C_3), (C_6,C_7), (C_{14},C_{15}), \ldots$ in the power series expansion of its inverse 
\[A^{-1} = \sum_{i \geq 0} C_i(x)  \cdot {(x^d)}^i, \qquad \text{with}\quad 
C_i \in \mathcal{M}_n \left(\K[x]_{<d} \right), \]
where $\mathcal{M}_n \left(\K[x]_{<d} \right) $ denotes the set of $n \times n$ matrices whose entries are polynomials in $\K[x]$ of degree less than~$d$.
                                                  
For instance, two extreme cases are ($i$) if $d=1$ and $A = I_n - xM$, with $M\in\mathcal{M}_n(\K)$, then $C_i = M^i$ and the high-order components can be computed as in~\S\ref{sec:powmat}; ($ii$) $n=1$, then the problem reduces to the one we solved in~\S\ref{ssec:msb-first}.   
Storjohann~\cite[\S5]{Storjohann02} proposed an algorithm for arbitrary $d$ and $n$, extrapolating between the two particular cases, with complexity 
$O(\MM(n,d) \, \log(N))$, where $\MM(n,d)$ denotes the arithmetic complexity 
of the product in $\mathcal{M}_n \left(\K[x]_{<d} \right)$.

Storjohann's algorithm relies on the following identities: 
\[\begin{cases}
	C_{2^i-2} & = \; - \; \left[ \; \big( C_{2^{i-1}-2} + C_{2^{i-1}-1} \cdot x^d \big) \cdot \big[ A\cdot C_{2^{i-1}-2} \big]_{d-1}^{2d-1} \; \right]_{d-1}^{2d-1}, \\                       \\
	C_{2^i-1} & = \; - \; \left[ \; \big( C_{2^{i-1}-2} + C_{2^{i-1}-1} \cdot x^d \big) \cdot \big[ A\cdot C_{2^{i-1}-1} \big]_{d-1}^{2d-1} \; \right]_{d-1}^{2d-1}.
\end{cases}\]       
Here, for a polynomial matrix $B=\sum_i B_i x^i$ with $B_i\in\mathcal{M}_n(\K)$, we use the notation $[B]_{d-1}^{2d-1}$ to denote the matrix $\sum_{i=0}^{d-1} B_{d+i} x^i$.

Thus, to compute $C_N$ say when $N$ is of the form $2^k-1$, this algorithm
uses     
$\sim 6 \, \MM(n,d) \log (N)$ operations in $\K$ if polynomial products are used, or  
$\sim 4 \, \MM(n,d) \log (N)$ operations in $\K$ if middle product techniques
are used for the outmost products.   
Using a matrix adaptation of Algorithm~\ref{alg:baset}, we can lower this to     
$\sim 2 \, \MM(n,d) \log (N)$ operations in $\K$.   

Note that Nuel and Dumas compared in~\cite{NuDu13} Fiduccia's and Storjohann's
algorithms (in the scalar case), but only in the under the specific assumption
that naive polynomial multiplication is used, that is $\M(d) = O(d^2)$.


\section{Analysis under the FFT multiplication model}\label{sec:FFT}

In this section, we specialize, optimize and analyze the generic
Algorithm~\ref{alg:base} to the FFT setting, in which polynomial
multiplications are assumed to be performed using the discrete Fourier
transform (DFT), and its inverse.

In order to do this, we will assume that the base ring $R$ possesses roots of
unity of sufficiently high order. To simplify the exposition, the ring $R$
will be supposed to be a field, but the arguments also apply without this
assumption, modulo some technical complications, see~\cite[\S 8.2]{GaGe13}.

\subsection{Discrete Fourier Transform for polynomial products}  \label{sec:DFTmul}
Let $\mathbb{K}$ be a field with a primitive $n$-th root $\omega_n$ of unity. Let $A \in \mathbb{K}[x]$ be a polynomial of degree at most $d\le n-1$.
The DFT $\widehat{A}$ of $A$ is defined by
\begin{align*}
    \widehat{A}_y := A(\omega_n^{-y}) = \sum_{i=0}^{n-1} A_i \omega_n^{-yi}\qquad \text{for } y=0,1,\dotsc, n-1.
\end{align*}
Here, $A_y = 0$ for $y > d$. 
It is classical that the DFT map is an invertible $\K$-linear transform from $\K^n$ to itself, and that the polynomial $A$ can be retrieved from its DFT $\widehat{A}$ using the formulas
\begin{align*}
    A_i = \frac1{n} \sum_{y=0}^{n-1} \widehat{A}_y \omega_n^{yi}\qquad \text{for } i=0,1,\dotsc, n-1.
\end{align*}

For computing the polynomial multiplication $C(x)=A(x)B(x)$ for given $A(x), B(x)\in\mathbb{K}[x]$ of degree at most $d$,
it is sufficient to compute the DFT of $C(x)$ for $n\ge 2d+1$.
Since $\widehat{C}_y = C(\omega_n^{-y}) = A(\omega_n^{-y})B(\omega_n^{-y}) = \widehat{A}_y\widehat{B}_y$, the polynomial $C(x)$ can be computed using two DFTs and one inverse DFT.

Let $\E(n)$ be an arithmetic complexity for computing a DFT of length $n$. Then the cost of polynomial multiplication in $\K[x]$ is governed by 
\[\M(d) = 3 \, \E(2d) + O(d).\] 

\subsection{Fast Fourier Transform}
In this subsection, we briefly recall the Fast Fourier Transform (FFT), which gives the quasi-linear estimate $\E(n)=O(n\log n)$.

Assume $n$ is even. Then,   for $y=0,1,\dotsc,n/2-1$ we have
\begin{align*}
    \widehat{A}_{y} &= \sum_{i=0}^{n/2-1} A_{2i} \omega_n^{-y(2i)} + \sum_{i=0}^{n/2-1} A_{2i+1} \omega_n^{-y(2i+1)}\\
                    &= \sum_{i=0}^{n/2-1} A_{2i} \omega_{n/2}^{-yi} + \omega_n^{-y}\cdot \sum_{i=0}^{n/2-1} A_{2i+1} \omega_{n/2}^{-yi}\\
                    &= \widehat{A^{\rm e}}_y + \omega_n^{-y} \widehat{A^{\rm o}}_y
\end{align*}
 where $A^{\rm e}(x):=\sum_{i=0}^{n/2} A_{2i} x^i$ and $A^{\rm o}(x):=\sum_{i=0}^{n/2} A_{2i+1} x^i$.

Similarly, we have  $\widehat{A}_{n/2+y} = \widehat{A}^e_y - \omega_n^{-y} \widehat{A}^o_y$.
We therefore obtain the following matrix equation
\begin{align}
\begin{bmatrix}
\widehat{A}_y\\
\widehat{A}_{n/2+y}\\
\end{bmatrix}
=
\begin{bmatrix}
1&1\\
1&-1
\end{bmatrix}                                                                                                                                                                                                                                
\begin{bmatrix}
1&0\\
0&\omega_N^{-y}
\end{bmatrix}
\begin{bmatrix}
\widehat{A}^{\mathrm{e}}_{y}\\
\widehat{A}^{\mathrm{o}}_{y}\\
\end{bmatrix}
\qquad\text{for } y=0,1,\dotsc,n/2-1.
\label{eq:FFT}
\end{align}    
Thus, computing a DFT in size $n$ reduces to two DFTs in size $n/2$. 
More precisely, $\E(n) \le 2 \, \E(n/2) + (3/2)n$. 
If $n$ is a power of two, $n=2^k$, and if the field~$\K$ contains a primitive $2^k$-th root of unity (as is the case for instance when $\K=\C$,
$\K=\mathbb{F}_p$ for a prime number $p$
satisfying $2^k\mid p-1$), this reduction can be repeated $k=\log n$ times, and it yields the estimate $\E(n) = \frac32 n\log n$.
The corresponding algorithm is called the \emph{decimation-in-time} Cooley--Tukey fast Fourier transform~\cite{CoTu65}, see also~\cite[\S2]{Bernstein08}.   

By the arguments of \S\ref{sec:DFTmul}, we conclude that polynomial multiplication in $\K[x]$ can be performed in arithmetic complexity
\[\M(d) = 9 \, d\log d + O(d).\]

\subsection{Efficiently doubling the length of a DFT}  \label{sec:doubling}
In the FFT setting, it is useful for many applications to compute efficiently      
a DFT of length $2n$ starting from DFT of length~$n$.

Assume $n\ge d+1$ and we have at our disposal the DFT $\widehat{A}$ of $A$, of
length~$n$. Assume that we want to compute the DFT $\widehat{A}^{(2n)}$ of
length $2n$.

The simplest algorithm is to apply the inverse DFT of length~$n$ to
obtain~$A$, and then to apply the DFT of length $2n$ to $A$. This costs $\E(n)+\E(2n)$
arithmetic operations, that is $\frac92 n \log n + 3n$ operations in $\K$.    

This algorithm can be improved using the following formulas
\begin{align*}
    \widehat{A}_{2y}^{(2n)} &= \sum_{i=0}^{2n-1} A_i \omega_{2n}^{-2yi} = \sum_{i=0}^{n-1} A_i \omega_{n}^{-yi} = \widehat{A}_y,\\
    \widehat{A}_{2y+1}^{(2n)} &= \sum_{i=0}^{2n-1} A_i \omega_{2n}^{-(2y+1)i} = \sum_{i=0}^{n-1} \omega_{2n}^{-i}A_i \omega_{n}^{-yi} = \widehat{B}_y,
\end{align*}
where $B_i:=\omega_{2n}^{-i}A_i$ for $i=0,1,\dotsc,n-1$. 
We obtain Algorithm~\ref{alg:up} with arithmetic complexity $ 2 \, \E(n) + n$, \emph{i.e.} $3 n \log n + n$,~\cite[\S12.8]{Bernstein08}, see also~\cite{Bernstein04,Mezzarobba10}\footnote{This ``FFT doubling'' trick is sometimes attributed to R.~Kramer (2004), but we have not been able to locate Kramer's paper.}.
Compared with the direct algorithm, the gain is roughly a factor of $3/2$.
\begin{algorithm}[t]
\begin{algorithmic}[1]
\Function {{\sf UP}}{$\widehat{A}$}
\State $A \gets \mathrm{IDFT}_{n}(\widehat{A})$
\State $B_i \gets \omega_{2n}^{-i} A_i$ \qquad for $i=0,1,\dotsc,n-1$
\State $\widehat{B}\gets \mathrm{DFT}_n(B)$
\State $\widehat{A}^{(2n)}_{2y} \gets \widehat{A}_y$ for $y = 0,1,\dotsc,n-1$
\State $\widehat{A}^{(2n)}_{2y+1} \gets \widehat{B}_y$ for $y = 0,1,\dotsc,n-1$
\State \Return $\widehat{A}^{(2n)}$
\EndFunction
\end{algorithmic}
\caption{Doubling the length of a DFT}
\label{alg:up}
\end{algorithm}

\subsection{Algorithm~\ref{alg:base} in the FFT setting}  
Recall that our main objective is, given $P,Q$ in $\K[x]$ with $d=\deg(Q) > \deg(P)$, to compute the $N$-th coefficient $u_N$ in the series expansion of~$P/Q$.

Let $k$ be the minimum integer satisfying $2^k\ge 2d+1$. Assume that there exists a primitive $2^k$-th root of unity in $\mathbb{K}$.
In this case, we can employ an FFT-based polynomial multiplication in $\K[x]$.
In each iteration of Algorithm~\ref{alg:base}, it is sufficient to compute $P(x)Q(-x)$ and $Q(x)Q(-x)$.
Here, only two FFTs and two inverse FFT of length $2^k$ are needed since $\widehat{Q}^-_y = \widehat{Q}_{\bar{y}}$ for $Q^-(x) := Q(-x)$ where $\bar{y}:=y+2^{k-1}$ if $y<2^{k-1}$ and $\bar{y}:=y-2^{k-1}$ if $y\ge 2^{k-1}$.
Hence, the arithmetic complexity $\SS(d)$ for a single step in 
Algorithm~\ref{alg:base} satisfies $\SS(d)\le 4 \, \E(2^k) + O(2^k)$.

In the following we will show the improved estimate 
\[\SS(d)\le 4 \, \E(2^{k-1}) + O(2^k) . \]

Before entering the \textbf{while} loop in Algorithm~\ref{alg:base}, the DFTs $\widehat{P}$ and $\widehat{Q}$ of $P(x)$ and $Q(x)$ of length $2^k$ are computed, respectively.
Inside the \textbf{while} loop, $\widehat{P}$ and $\widehat{Q}$ are updated.
The recursive formula~\eqref{eq:FFT} for the decimation-in-time Cooley--Tukey FFT is equivalent to
\begin{align*}
\begin{bmatrix}
\widehat{A}^{\mathrm{e}}_y\\
\widehat{A}^{\mathrm{o}}_y\\
\end{bmatrix}
=
\frac12
\begin{bmatrix}
1&0\\
0&\omega_N^{y}
\end{bmatrix}
\begin{bmatrix}
1&1\\
1&-1
\end{bmatrix}
\begin{bmatrix}
\widehat{A}_y\\
\widehat{A}_{2^{k-1}+y}\\
\end{bmatrix}
\qquad\text{for } y=0,1,\dotsc,2^{k-1}.
\end{align*}
By using this formula, $\widehat{A}^{\rm e}$ (or $\widehat{A}^{\rm o}$) can be computed with $O(2^k)$ operations from $\widehat{A}$.
By using Algorithm~\ref{alg:up}, we obtain the updated $\widehat{P}$ from $\widehat{A}^{\rm e}$ or $\widehat{A}^{\rm o}$.
The algorithm is summarized in Algorithm~\ref{alg:FFT}.
In each step, $\mathsf{UP}$ is called twice. Hence, the total arithmetic complexity of Algorithm~\ref{alg:FFT} is 
\[(4 \, \E(2^{k-1}) + O(2^k)) \cdot \log N.\]    

When $d$ is of the form $2^\ell-1$\footnote{In the general case, it might be useful to use the Truncated Fourier Transform (TFT), which
smoothes the ``jumps'' in complexity exhibited by FFT algorithms~\cite{Hoeven04,HaRo10,Arnold13}.}, then one can take $k=\ell+1$
and the cost simplifies to
\[\T(N,d)= 4 \, \E(d) \log N + O(d \log N),\]    
or, equivalently
\[\T(N,d)= 6\, d \log d \log N + O(d \log N).\]    
The (striking) conclusion of this analysis is that, in the FFT setting, our 
(variant of the) algorithm for computing the $N$-th term of $P/Q$ uses much less operations than in the general case, namely
\begin{equation} \label{conv:FFT}
	\T(N, d) \sim \frac23 \, \M(d) \log N,
\end{equation}
while for a generic multiplication algorithm the cost is $\sim 2 \, \M(d) \log N$.   This proves Theorem~\ref{thm:main-FFT}.

Note that the complexity bound~\eqref{conv:FFT} compares favorably with
Fidducia's algorithm combined with the best algorithms for modular squaring.
For instance, Shoup's algorithm~\cite[\S7.3]{Shoup95} computes one modular
squaring in the FFT setting using $\sim \frac53 \, \M(d)$ arithmetic
operations, while Mih\u{a}ilescu's algorithm~\cite[Table~1]{Mihailescu08}
(based on Montgomery's algorithm~\cite{Montgomery85}) uses roughly $\sim
\frac{13}{12} \, \M(d)$ arithmetic operations. Our bound~\eqref{conv:FFT} is
better by a factor of~$2.5$ than Shoup's (comparatively simple) algorithm, and
by a factor of~$1.625$ than the (much more complex) algorithm by
Mih\u{a}ilescu.

Let us point out that all the other algorithms admit similarly fast versions
in the FFT setting. We will however not give them in full detail here, mainly
for space reasons.

\begin{algorithm}[t]
\begin{algorithmic}[1]
\State $\widehat{P} \gets \mathrm{DFT}_{2^k}(P)$
\State $\widehat{Q} \gets \mathrm{DFT}_{2^k}(Q)$
\While{$N \ge 1$}
\State $\widehat{U}_y \gets \widehat{P}_y\,\widehat{Q}_{\bar{y}}$ for $y=0,1,\dotsc, 2^k-1$
\If{$N$ is even}
\State $\widehat{U}^{\rm e}_y \gets (\widehat{U}_y + \widehat{U}_{y + 2^{k-1}})/2$ for $y=0,1,\dotsc, 2^{k-1}-1$
\State $\widehat{P} \gets \mathsf{UP}(\widehat{U}^{\rm e}_y)$
\Else
\State $\widehat{U}^{\rm o}_y \gets \omega_N^y(\widehat{U}_y - \widehat{U}_{y + 2^{k-1}})/2$ for $y=0,1,\dotsc, 2^{k-1}-1$
\State $\widehat{P} \gets \mathsf{UP}(\widehat{U}^{\rm o}_y)$
\EndIf
\State $\widehat{A}_y \gets \widehat{Q}_y\,\widehat{Q}_{\bar{y}}$ for $y=0,1,\dotsc, 2^{k-1}-1$
\State $\widehat{Q} \gets \mathsf{UP}(\widehat{A})$
\State $N \gets \lfloor N/2\rfloor$
\EndWhile
\State $P(0) \gets \sum_{y=0}^{2^k-1} \widehat{P}_y$
\State $Q(0) \gets \sum_{y=0}^{2^k-1} \widehat{Q}_y$
\State \Return $P(0)/Q(0)$
\end{algorithmic}
\caption{(\textcolor{red}{\textsf{OneCoeff-FFT}}) \; {\bf Input}: $P(x)$, $Q(x)$, $N$\; {\bf Output}: $[x^N]\, \frac{P(x)}{Q(x)}$}
\label{alg:FFT}
\end{algorithm}

\section{Conclusion}         \label{sec:conclusion}

We have proposed several algorithmic contributions to the classical field of
linearly recurrent sequences.

Firstly, we have designed a simple and fast algorithm for computing the $N$-th
term of a linearly recurrent sequence of order~$d$, using $\sim 2 \, \M(d)
\log N$ arithmetic operations, which is faster by a factor of~$1.5$ than the
state-of-the-art 1985 algorithm due to Fiduccia~\cite{Fiduccia85}. When
combined with FFT techniques, the algorithm has even better arithmetic
complexity $\sim \frac23 \, \M(d) \log N$ which is faster than the fastest
variant of Fiduccia's algorithm in the FFT setting by a factor of~1.625. The
new algorithms are based on a new method (Algorithm~\ref{alg:base}) for
computing the $N$-th coefficient of a rational power series.

Secondly, using algorithmic transposition techniques, we have derived from
Algorithm~\ref{alg:base} a new method (Algorithm~\ref{alg:baset}) for
computing simultaneously the coefficients of indices $N-d+1,\ldots,N$ in the
power series expansion of the reciprocal of a degree-$d$ polynomial, using
again $\sim 2 \, \M(d) \log N$ arithmetic operations. Using
Algorithm~\ref{alg:baset}, we have designed a new algorithm for computing the
remainder of $x^N$ modulo a given polynomial of degree~$d$, using $\sim 2 \,
\M(d) \log N$ arithmetic operations as well. This is better by a factor of 1.5
than the previous best algorithm for modular exponentiation, with an even
better speed-up in the FFT setting, as for Algorithm~\ref{alg:base}. Combined
with the basic idea of Fiduccia's algorithm, our new algorithm for modular
exponentiation yields a faster Fiduccia-like algorithm (by the aforementioned
constant factors) that computes a slice of $d$ consecutive terms (of indices
$N-d+1,\ldots,N$) of a linearly recurrent sequence of order~$d$ using $\sim 2
\, \M(d) \log N$ arithmetic operations.

Thirdly, we have discussed applications of the new algorithms to a few other
algorithmic problems, including powering of matrices, high-order lifting (a
basic brick for modern polynomial linear algebra algorithms) and the
computation of terms of linearly recurrent sequences when the recurrence has
roots with (high) multiplicities.

\smallskip As future work, we plan to investigate further the full power of
our technique. To which extent can it be generalized to larger classes of
power series? For instance, although it perfectly works for bivariate rational
power series $U(x,y)$, the corresponding method does not directly provide a
$O(\log N)$-algorithm for computing the $(N,N)$-th coefficient $u_{N,N}$, the
reason being that the $\log N$ new bivariate recurrences produced by the
Graeffe process do not have constant orders, as in the univariate case. This
is disappointing, but after all not surprising, because the generating
function of the sequence $(u_{n,n})_n$ is known to be algebraic, but not
rational anymore~\cite{Polya22}. As of today, no algorithm is known for
computing the $N$-th coefficient of an algebraic power series faster than in
the P-recursive case (P), namely in a number of ring operations almost linear
in~$\sqrt{N}$.

\medskip\noindent {\bf Acknowledgements.} 
Our special thanks go to Kevin Atienza, whose  
editorial on \href{https://discuss.codechef.com/t/rng-editorial/10068}{https://discuss.codechef.com} was our initial source of inspiration, and to 
Sergey Yurkevich, for his careful reading of a first draft of this work.
A.~Bostan was supported in part by
 \textcolor{magenta}{\href{https://specfun.inria.fr/chyzak/DeRerumNatura/}{DeRerumNatura}} ANR-19-CE40-0018.  \\  
R.~Mori was supported in part by JST PRESTO Grant $\#$JPMJPR1867 and JSPS KAKENHI Grant $\#$JP17K17711, $\#$JP18H04090 and $\#$JP20H04138.

\def\gathen#1{{#1}}


\begin{thebibliography}{10}

\bibitem{AKS04}
Manindra Agrawal, Neeraj Kayal, and Nitin Saxena.
\newblock P{RIMES} is in {P}.
\newblock {\em Ann. of Math. (2)}, 160(2):781--793, 2004.

\bibitem{AlSh92}
Jean-Paul Allouche and Jeffrey Shallit.
\newblock The ring of {$k$}-regular sequences.
\newblock {\em Theoret. Comput. Sci.}, 98(2):163--197, 1992.

\bibitem{Arnold13}
Andrew Arnold.
\newblock A new truncated {F}ourier transform algorithm.   
\newblock In {\em Proceedings of ISSAC'13}, pages 15--22. ACM Press, 2013.

\bibitem{Bernstein04}
D.~J. Bernstein.
\newblock Removing redundancy in high-precision {N}ewton iteration, 2004.
\newblock Preprint, \url{http://cr.yp.to/fastnewton.html}.

\bibitem{Bernstein08}
Daniel~J. Bernstein.
\newblock Fast multiplication and its applications.
\newblock In {\em Algorithmic number theory: lattices, number fields, curves
  and cryptography}, volume~44 of {\em Math. Sci. Res. Inst. Publ.}, pages
  325--384. Cambridge Univ. Press, 2008.

\bibitem{BiPa94}
Dario Bini and Victor~Y. Pan.
\newblock {\em Polynomial and matrix computations. {V}ol. 1}.
\newblock Progress in Theoretical Computer Science. Birkh\"{a}user Boston,
  Inc., Boston, MA, 1994.
\newblock Fundamental algorithms.

\bibitem{BSS99}
I.~Blake, G.~Seroussi, and N.~Smart.
\newblock {\em Elliptic curves in cryptography}, volume 265 of {\em London
  Math. Soc. Lecture Note Ser.}
\newblock Cambridge University Press, 1999.

\bibitem{BoLeSc03}
A.~Bostan, G.~Lecerf, and {\'E}.~Schost.
\newblock Tellegen's principle into practice.
\newblock In {\em Proceedings of ISSAC'03}, pages 37--44. ACM Press, 2003.

\bibitem{Bostan20}
Alin Bostan.
\newblock Computing the {$N$}-th {T}erm of a $q$-{H}olonomic {S}equence.
\newblock In {\em Proceedings of ISSAC'20}, pages 46--53. ACM Press, 2020.

\bibitem{BoCaChDu19}
Alin Bostan, Xavier Caruso, Gilles Christol, and Philippe Dumas.
\newblock Fast coefficient computation for algebraic power series in positive
  characteristic.
\newblock In {\em Proceedings of the {T}hirteenth {A}lgorithmic {N}umber
  {T}heory {S}ymposium}, volume~2 of {\em Open Book Ser.}, pages 119--135.
  Math. Sci. Publ., Berkeley, CA, 2019.

\bibitem{BoChDu16}
Alin Bostan, Gilles Christol, and Philippe Dumas.
\newblock Fast computation of the {$N$}th term of an algebraic series over a
  finite prime field.
\newblock In {\em Proceedings of ISSAC'16}, pages 119--126. ACM Press, 2016.

\bibitem{BoGaSc07}
Alin Bostan, Pierrick Gaudry, and \'{E}ric Schost.
\newblock Linear recurrences with polynomial coefficients and application to
  integer factorization and {C}artier-{M}anin operator.
\newblock {\em SIAM J. Comput.}, 36(6):1777--1806, 2007.

\bibitem{BGY80}
Richard~P. Brent, Fred~G. Gustavson, and David Y.~Y. Yun.
\newblock Fast solution of {T}oeplitz systems of equations and computation of
  {P}ad\'{e} approximants.
\newblock {\em J. Algorithms}, 1(3):259--295, 1980.

\bibitem{CDJPS07}
Neil Calkin, Jimena Davis, Kevin James, Elizabeth Perez, and Charles Swannack.
\newblock Computing the integer partition function.
\newblock {\em Math. Comp.}, 76(259):1619--1638, 2007.

\bibitem{CeMiPi87}
L.~Cerlienco, M.~Mignotte, and F.~Piras.
\newblock Suites récurrentes linéaires. {P}ropriétés algébriques et  arithmétiques.
\newblock {\em L'Enseignement Mathématique}, 33:67--108, 1987.

\bibitem{ChCh88}
D.~V. Chudnovsky and G.~V. Chudnovsky.
\newblock Approximations and complex multiplication according to {R}amanujan.
\newblock In {\em Ramanujan revisited ({U}rbana-{C}hampaign, {I}ll., 1987)},
  pages 375--472. Academic Press, Boston, MA, 1988.

\bibitem{CoTu65}
James~W. Cooley and John~W. Tukey.
\newblock An algorithm for the machine calculation of complex {F}ourier series.
\newblock {\em Math. Comp.}, 19:297--301, 1965.

\bibitem{CuHa89}
Paul Cull and James~L. Holloway.
\newblock Computing {F}ibonacci numbers quickly.
\newblock {\em Inform. Process. Lett.}, 32(3):143--149, 1989.

\bibitem{Rocquigny1899}
G.~de~Rocquigny.
\newblock Question 1541. [I25a].
\newblock {\em L'Intermédiaire des mathématiciens}, 6:148, 1899.
                   
\bibitem{Dickson1919}
Leonard~Eugene Dickson.
\newblock {\em History of the theory of numbers. {V}ol. {I}: {D}ivisibility and
  primality}.
\newblock Publication No. 256. Washington: Carnegie Institution of Washington, Vol. 1, 1919.

\bibitem{Dijkstra79}
E.W. Dijkstra.
\newblock In honour of {F}ibonacci.
\newblock In {\em Program Construction. Lecture Notes in Computer Science 69,
  ed. F.L. Bauer et al.}, pages 49--50. Springer, Berlin, Heidelberg, 1979.

\bibitem{Er83}
M.~C. Er.
\newblock Computing sums of order-{$k$} {F}ibonacci numbers in log time.
\newblock {\em Inform. Process. Lett.}, 17(1):1--5, 1983.

\bibitem{Er86}
M.~C. Er.
\newblock A formal derivation of an {$O({\rm log}\,n)$} algorithm for computing
  {F}ibonacci numbers.
\newblock {\em J. Inform. Optim. Sci.}, 7(1):9--15, 1986.

\bibitem{Er88}
M.~C. Er.
\newblock An {$O(k^2\log (n/k))$} algorithm for computing generalized
  order-{$k$} {F}ibonacci numbers with linear space.
\newblock {\em J. Inform. Optim. Sci.}, 9(3):343--353, 1988.

\bibitem{Fiduccia82}
Charles~M. Fiduccia.
\newblock The $n$-th power of a companion matrix: Fast solutions to linear recurrences.
\newblock 20th Proceedings of the Annual Allerton Conference on Communication, Control and Computing, pages 934--940, 1982.

\bibitem{Fiduccia85}
Charles~M. Fiduccia.
\newblock An efficient formula for linear recurrences.
\newblock {\em SIAM J. Comput.}, 14(1):106--112, 1985.

\bibitem{GaGe13}
J.~\gathen{von zur} Gathen and J.~Gerhard.
\newblock {\em Modern computer algebra}.
\newblock Cambridge Univ. Press, third edition, 2013.

\bibitem{Giesbrecht95}
Mark Giesbrecht.
\newblock Nearly optimal algorithms for canonical matrix forms.
\newblock {\em SIAM J. Comput.}, 24(5):948--969, 1995.

\bibitem{GiGo05}
Kenneth~J. Giuliani and Guang Gong.
\newblock New {LFSR}-{B}ased {C}ryptosystems and the {T}race {D}iscrete {L}og
  {P}roblem ({T}race-{DLP}).
\newblock In Tor Helleseth, Dilip Sarwate, Hong-Yeop Song, and Kyeongcheol
  Yang, editors, {\em Sequences and Their Applications - SETA 2004}, pages
  298--312, Berlin, Heidelberg, 2005. Springer Berlin Heidelberg.

\bibitem{GiGo06}
Kenneth~J. Giuliani and Guang Gong.
\newblock A new algorithm to compute remote terms in special types of
  characteristic sequences.
\newblock In {\em Sequences and their applications---{SETA} 2006}, volume 4086
  of {\em Lecture Notes in Comput. Sci.}, pages 237--247. Springer, Berlin,
  2006.

\bibitem{GoHa99}
Guang Gong and Lein Harn.
\newblock Public-key cryptosystems based on cubic finite field extensions.
\newblock {\em IEEE Trans. Inform. Theory}, 45(7):2601--2605, 1999.

\bibitem{GoHaWu01}
Guang Gong, Lein Harn, and Huapeng Wu.
\newblock The {GH} public-key cryptosystem.
\newblock In {\em Selected areas in cryptography}, volume 2259 of {\em Lecture
  Notes in Comput. Sci.}, pages 284--300. Springer, Berlin, 2001.

\bibitem{GHMV13}
Santos Gonz\'{a}lez, Lloren\c{c} Huguet, Consuelo Mart\'{\i}nez, and Hugo
  Villafa\~{n}e.
\newblock Discrete logarithm like problems and linear recurring sequences.
\newblock {\em Adv. Math. Commun.}, 7(2):187--195, 2013.

\bibitem{GrLe80}
David Gries and Gary Levin.
\newblock Computing {F}ibonacci numbers (and similarly defined functions) in
  log time.
\newblock {\em Inform. Process. Lett.}, 11(2):68--69, 1980.

\bibitem{HaQuZi04}
Guillaume Hanrot, Michel Quercia, and Paul Zimmermann.
\newblock The middle product algorithm. {I}.
\newblock {\em Appl. Algebra Engrg. Comm. Comput.}, 14(6):415--438, 2004.

\bibitem{Harvey14}
David Harvey.
\newblock Counting points on hyperelliptic curves in average polynomial time.
\newblock {\em Ann. of Math. (2)}, 179(2):783--803, 2014.

\bibitem{HaRo10}
David Harvey and Daniel~S. Roche.
\newblock An in-place truncated {F}ourier transform and applications to
  polynomial multiplication.   
\newblock In {\em Proceedings of ISSAC'10}, pages 325--329. ACM Press, 2010.

\bibitem{Hoeven04}
Joris van~der Hoeven.
\newblock The truncated {F}ourier transform and applications.
\newblock In {\em Proceedings of ISSAC'04}, pages 290--296. ACM Press, 2004.


\bibitem{Holloway88}
J.~L. Holloway.
\newblock Algorithms for computing {F}ibonacci numbers quickly.
\newblock MSc Thesis. Oregon State Univ. 1988.

\bibitem{Householder59}
Alston~S. Householder.
\newblock Dandelin, {L}oba\v{c}evski\u{\i}, or {G}raeffe?
\newblock {\em Amer. Math. Monthly}, 66:464--466, 1959.

\bibitem{HyMeScSt19}
Seung~Gyu Hyun, Stephen Melczer, \'{E}ric Schost, and Catherine St-Pierre.
\newblock Change of basis for {$\mathfrak{m}$}-primary ideals in one and two
  variables.
\newblock In {\em Proceedings of ISSAC'19}, pages 227--234. ACM Press, 2019.

\bibitem{HyMeSt19}
Seung~Gyu Hyun, Stephen Melczer, and Catherine St-Pierre.
\newblock A fast algorithm for solving linearly recurrent sequences.
\newblock {\em {ACM} Communications in Computer Algebra}, page 100–103, 2019.

\bibitem{Keller85}
Walter Keller-Gehrig.
\newblock Fast algorithms for the characteristic polynomial.
\newblock {\em Theoret. Comput. Sci.}, 36(2-3):309--317, 1985.

\bibitem{Khomovsky18}
Dmitry~I. Khomovsky.
\newblock Efficient computation of terms of linear recurrence sequences of any
  order.
\newblock {\em Integers}, 18:Paper No. A39, 12, 2018.

\bibitem{Knuth69}
Donald~E. Knuth.
\newblock {\em The art of computer programming. {V}ol. 2: {S}eminumerical
  algorithms}.
\newblock Addison-Wesley Publishing Co., Reading, Mass.-London-Don Mills, Ont,  
first edition, 1969.

\bibitem{Knuth81}
Donald~E. Knuth.
\newblock {\em The art of computer programming. {V}ol. 2}.
\newblock Addison-Wesley Publishing Co., Reading, Mass., second edition, 1981.
\newblock Seminumerical algorithms, Addison-Wesley Series in Computer Science
  and Information Processing.

\bibitem{Knuth81e}
Donald~E. Knuth.
\newblock The {L}ast {W}hole {E}rrata {C}atalog, 1981.
\newblock Department of Computer Science, Stanford University, Report. No.
  STAN-CS-81-868,
  \url{http://infolab.stanford.edu/pub/cstr/reports/cs/tr/81/868/CS-TR-81-868.pdf}.

\bibitem{Lecerf10}
Gr\'{e}goire Lecerf.
\newblock New recombination algorithms for bivariate polynomial factorization
  based on {H}ensel lifting.
\newblock {\em Appl. Algebra Engrg. Comm. Comput.}, 21(2):151--176, 2010.

\bibitem{LuZi08}
Katharina L\"{u}rwer-Br\"{u}ggemeier and Martin Ziegler.
\newblock On faster integer calculations using non-arithmetic primitives.
\newblock In {\em Unconventional computation}, volume 5204 of {\em Lecture
  Notes in Comput. Sci.}, pages 111--128. Springer, Berlin, 2008.

\bibitem{MaRe84}
Alain~J. Martin and Martin Rem.
\newblock A presentation of the {F}ibonacci algorithm.
\newblock {\em Inform. Process. Lett.}, 19(2):67--68, 1984.

\bibitem{NaPa13}
J.~M. McNamee and V.~Y. Pan.
\newblock {\em Numerical methods for roots of polynomials. {P}art {II}},
  volume~16 of {\em Studies in Computational Mathematics}.
\newblock Elsevier/Academic Press, Amsterdam, 2013.

\bibitem{Mezzarobba10}
Marc Mezzarobba.
\newblock Num{G}fun: a package for numerical and analytic computation and
  {D}-finite functions.         
\newblock In {\em Proceedings of ISSAC'10}, pages 139--146. ACM Press, 2010.

\bibitem{Mihailescu08}
Preda Mih\u{a}ilescu.
\newblock Fast convolutions meet {M}ontgomery.
\newblock {\em Math. Comp.}, 77(262):1199--1221, 2008.

\bibitem{MiBr66}
J.~C.~P. Miller and D.~J.~Spencer Brown.
\newblock An algorithm for evaluation of remote terms in a linear recurrence
  sequence.
\newblock {\em Computer Journal}, 9:188--190, 1966.

\bibitem{Montgomery85}
Peter~L. Montgomery.
\newblock Modular multiplication without trial division.
\newblock {\em Math. Comp.}, 44(170):519--521, 1985.

\bibitem{NuDu13}
Gregory Nuel and Jean-Guillaume Dumas.
\newblock Sparse approaches for the exact distribution of patterns in long
  state sequences generated by a {M}arkov source.
\newblock {\em Theoret. Comput. Sci.}, 479:22--42, 2013.

\bibitem{Pan87}
V.~Pan.
\newblock Algebraic complexity of computing polynomial zeros.
\newblock {\em Comput. Math. Appl.}, 14(4):285--304, 1987.

\bibitem{Pan97}
Victor~Y. Pan.
\newblock Solving a polynomial equation: some history and recent progress.
\newblock {\em SIAM Rev.}, 39(2):187--220, 1997.

\bibitem{PaSt73}
Michael~S. Paterson and Larry~J. Stockmeyer.
\newblock On the number of nonscalar multiplications necessary to evaluate
  polynomials.
\newblock {\em SIAM J. Comput.}, 2:60--66, 1973.

\bibitem{Pettorossi80}
Alberto Pettorossi.
\newblock Derivation of an {$O(k^{2}\,{\rm log}\,n)$} algorithm for computing
  order-{$k$} {F}ibonacci numbers from the {$O(k^{3}\,{\rm log}\,n)$} matrix
  multiplication method.
\newblock {\em Inform. Process. Lett.}, 11(4-5):172--179, 1980.

\bibitem{Polya22}
G.~P\'olya.
\newblock Sur les s\'eries enti\`eres, dont la somme est une fonction
  alg\'ebrique.
\newblock {\em Enseignement Math.}, 22:38--47, 1921/1922.

\bibitem{Poorten89}
A.~J. van~der Poorten.
\newblock Some facts that should be better known, especially about rational
  functions.
\newblock In {\em Number theory and applications ({B}anff, {AB}, 1988)}, volume
  265 of {\em NATO Adv. Sci. Inst. Ser. C Math. Phys. Sci.}, pages 497--528.
  Kluwer Acad. Publ., Dordrecht, 1989.

\bibitem{PrTa89}
Marco Protasi and Maurizio Talamo.
\newblock On the number of arithmetical operations for finding {F}ibonacci
  numbers.
\newblock {\em Theoret. Comput. Sci.}, 64(1):119--124, 1989.

\bibitem{Ranum11}
Arthur Ranum.
\newblock The general term of a recurring series.
\newblock {\em Bull. Amer. Math. Soc.}, 17(9):457--461, 1911.

\bibitem{RaEi04}
B.~Ravikumar and G.~Eisman.
\newblock Weak minimization of {DFA}---an algorithm and applications.
\newblock {\em Theoret. Comput. Sci.}, 328(1-2):113--133, 2004.
             
\bibitem{REMLP1900}
Rosace, E.-B. Escott, E.~Malo, C.-A. Laisant, and G.~Picou.
\newblock Answers to {Q}uestion 1541. [{I}25a] asked by {G}. de {R}ocquigny.
\newblock {\em L'Intermédiaire des mathématiciens}, 7:172--177, 1900.

\bibitem{Schonhage00}
Arnold Sch\"{o}nhage.
\newblock Variations on computing reciprocals of power series.
\newblock {\em Inform. Process. Lett.}, 74(1-2):41--46, 2000.

\bibitem{Shortt78}
Joseph Shortt.
\newblock An iterative program to calculate {F}ibonacci numbers in {$O({\rm
  log}\,n)$} arithmetic operations.
\newblock {\em Inform. Process. Lett.}, 7(6):299--303, 1978.

\bibitem{Shoup91}
Victor Shoup.
 \newblock A fast deterministic algorithm for factoring polynomials over finite
   fields of small characteristic.
\newblock In {\em Proceedings of ISSAC'91}, pages 14--21. ACM Press, 1991.

\bibitem{Shoup95}
Victor Shoup.
\newblock A new polynomial factorization algorithm and its implementation.
\newblock {\em J. Symbolic Comput.}, 20(4):363--397, 1995.

\bibitem{Storjohann02}
A.~Storjohann.
\newblock High-order lifting.
\newblock In {\em Proceedings of ISSAC'02}, pages 246--254. ACM Press, 2002.

\bibitem{Strassen76}
V.~Strassen.
\newblock Einige {R}esultate über {B}erechnungskomplexität.
\newblock {\em Jahresbericht der Deutschen Mathematiker-Vereinigung},
  78(1):1--8, 1976/77.

\bibitem{Strassen74}
Volker Strassen.
\newblock Polynomials with rational coefficients which are hard to compute.
\newblock {\em SIAM J. Comput.}, 3:128--149, 1974.

\bibitem{Takahashi00}
Daisuke Takahashi.
\newblock A fast algorithm for computing large {F}ibonacci numbers.
\newblock {\em Inform. Process. Lett.}, 75(6):243--246, 2000.

\bibitem{Urbanek80}
Friedrich~J. Urbanek.
\newblock An {$O({\rm log}$} {$n)$} algorithm for computing the {$n$}th element
  of the solution of a difference equation.
\newblock {\em Inform. Process. Lett.}, 11(2):66--67, 1980.

\bibitem{WiSh80}
Thomas~C. Wilson and Joseph Shortt.
\newblock An {$O({\rm log}$} {$n)$} algorithm for computing general order-{$k$}
  {F}ibonacci numbers.
\newblock {\em Inform. Process. Lett.}, 10(2):68--75, 1980.

\end{thebibliography}
\end{document}